\documentclass[11pt]{article}
\usepackage{graphicx,amssymb,amsmath}

\usepackage[linesnumbered, ruled]{algorithm2e}
\usepackage{amsfonts}
\usepackage{amsmath}
\usepackage{booktabs}
\usepackage{cite}
\usepackage{color}
\usepackage{enumerate}
\usepackage{float}
\usepackage{gensymb}
\usepackage{hyperref}
\usepackage{mathrsfs}
\usepackage{subfig}
\usepackage{tabularx}
\usepackage{fullpage}
\usepackage[top=0.8in, bottom=0.9in, left=0.9in, right=0.9in]{geometry}

\graphicspath{{./images/}}

\def\calC{\mathcal{C}}
\def\calI{\mathcal{I}}
\def\calL{\mathcal{L}}
\def\calU{\mathcal{U}}
\def\calD{\mathcal{D}}
\def\calK{\mathcal{K}}
\def\calR{\mathcal{R}}
\def\calQ{\mathcal{Q}}
\def\calF{\mathcal{F}}
\def\calA{\mathcal{A}}

\newtheorem{observation}{Observation}

\newtheorem{lemma}{Lemma}
\newtheorem{theorem}{Theorem}

\newtheorem{subproblem}{Subproblem}
\newenvironment{proof}{\noindent {\textbf{Proof:}}\rm}{\hfill $\Box$
\rm\bigskip}

\title{Reverse Shortest Path Problem for Unit-Disk Graphs\thanks{This research was supported in part by NSF under Grant CCF-2005323. Preliminary results of this paper appeared in {\em Proceedings of the 17th Algorithms and Data Structures Symposium (WADS 2021)} and {\em Proceedings of the 16th International Conference and Workshops on Algorithms and Computation (WALCOM 2022)}.}}

\author{
Haitao Wang\thanks{Department of Computer Science,
Utah State University, Logan, UT 84322, USA. {\tt haitao.wang@usu.edu}}
\and
Yiming Zhao\thanks{Corresponding author. Department of Computer Science,
Utah State University, Logan, UT 84322, USA. {\tt yiming.zhao@usu.edu}}
}

\begin{document}
\pagestyle{plain}
\date{}

\thispagestyle{empty}
\maketitle

 \vspace{-0.2in}
\begin{abstract}
Given a set $P$ of $n$ points in the plane, the unit-disk graph $G_{r}(P)$ with respect to a parameter $r$ is an undirected graph whose vertex set is $P$ such that an edge connects two points $p, q \in P$ if the Euclidean distance between $p$ and $q$ is at most $r$ (the weight of the edge is 1 in the unweighted case and is the distance between $p$ and $q$ in the weighted case). 
Given a value $\lambda>0$ and two points $s$ and $t$ of $P$, we consider the following {\em reverse shortest path problem}: computing the smallest $r$ such that the shortest path length between $s$ and $t$ in $G_r(P)$ is at most $\lambda$. 
In this paper, we present an algorithm of $O(\lfloor\lambda\rfloor \cdot n \log n)$ time and another algorithm of $O(n^{{5}/{4}} \log^{7/4} n)$ time for the unweighted case, as well as an  $O(n^{{5}/{4}} \log^{5/2} n)$ time algorithm for the weighted case. We also consider the $L_1$ version of the problem where the distance of two points is measured by the $L_1$ metric; we solve the problem in $O(n \log^3 n)$ time for both the unweighted and weighted cases.

\end{abstract}

\section{Introduction}
\label{sec:introduction}

Given a set $P$ of $n$ points in the plane and a parameter $r$, the
{\em unit-disk graph} $G_{r}(P)$ is an undirected graph whose vertex set is
$P$ such that an edge connects two points $p, q \in P$ if the
(Euclidean) distance between $p$ and $q$ is at most $r$.
The weight of each edge of $G_{r}(P)$ is defined to be one in the {\em unweighted} case 
and is defined to the distance between the two vertices of the edge in the {\em
weighted} case. Alternatively, $G_r(P)$ can be viewed as the
intersection graph of the set of congruous disks centered at the
points of $P$ with radii equal to $r/2$, i.e., two vertices are
connected if their disks intersect. The \emph{length} of a
path in $G_{r}(P)$ is the sum of the weights of the edges of the path.

Computing shortest paths in unit-disk graphs
with different distance metrics and different weights assigning
methods has been extensively studied,
e.g.,~\cite{ref:CabelloSh15,ref:ChanAl16,ref:ChanAp18,ref:GaoWe05,ref:KaplanDy17,ref:RodittyOn11,ref:WangNe20}.
Although a unit-disk graph may have $\Omega(n^2)$ edges, geometric properties allow to solve the single-source-shortest-path problem (SSSP) in sub-quadratic time. Roditty and Segal~\cite{ref:RodittyOn11} first proposed an algorithm of $O(n^{4/3 + \epsilon})$ time for  unit-disk graphs for both unweighted and weighted cases, for any $\epsilon>0$. Cabello and Jej\v{c}i\v{c}~\cite{ref:CabelloSh15} gave an algorithm of $O(n \log n)$ time for the unweighted case. Using a dynamic data structure for bichromatic closest pairs~\cite{ref:AgarwalVe99}, they also solved the weighted case in $O(n^{1 + \epsilon})$ time~\cite{ref:CabelloSh15}. Chan and Skrepetos~\cite{ref:ChanAl16} gave an $O(n)$ time algorithm for the unweighted case, assuming that all points of $P$ are presorted. Kaplan et al.~\cite{ref:KaplanDy17} developed a new randomized result for the dynamic bichromatic closest pair problem; applying the new result to the algorithm of~\cite{ref:CabelloSh15} leads to an $O(n \log^{12 + o(1)} n)$ expected time randomized algorithm for the weighted case. Recently, Wang and Xue~\cite{ref:WangNe20} proposed a new algorithm that solves the weighted case in $O(n \log^2 n)$ time. Somce approximation algorithms for the  problem have also been developed~\cite{ref:GaoWe05,ref:ChanAp18,ref:WangNe20}.

The $L_1$ version of the SSSP problem has also been studied, where the distance of two points in the plane is measured under the $L_1$ metric when defining $G_r(P)$. Note that in the $L_1$ version a ``disk'' is a diamond. The SSSP algorithms of \cite{ref:CabelloSh15,ref:ChanAl16} for the $L_2$ unweighted version can be easily adapted to the $L_1$ unweighted version. Wang and Zhao~\cite{ref:WangAn21} recently solved the $L_1$ weighted case in $O(n \log n)$ time. It is known that $\Omega(n\log n)$ is a lower bound for the SSSP problem in both $L_1$ and $L_2$ versions~\cite{ref:CabelloSh15,ref:WangAn21}. Hence, the SSSP problem in the $L_1$ weighted/unweighted case as well as in the $L_2$ unweighted case has been solved optimally.


In this paper, we consider the following \emph{reverse shortest path} (RSP)
problem. In addition to $P$, given a value $\lambda > 0$ and two points $s, t
\in P$, the problem is to compute the smallest value $r$ such that the distance between $s$ and $t$ in $G_{r}(P)$ is at most $\lambda$. There are four cases for the RSP problem depending on whether $L_1$ or $L_2$ metric is considered and whether the unit-disk graphs are weighted or not. Throughout the paper, we let $r^*$ denote the optimal value $r$ for any case. The goal is therefore to compute $r^*$.


Observe that $r^*$ must be equal to the distance of two points in $P$ in any case (i.e., $L_1$, $L_2$, weighted, unweighted). In light of this observation, Cabello and Jej\v{c}i\v{c}~\cite{ref:CabelloSh15} mentioned a straightforward solution that can compute $r^*$ in $O(n^{4/3} \log^3 n)$ time for both the unweighted and the weighted cases in the $L_2$ metric, by using the distance selection algorithm of Katz and Sharir~\cite{ref:KatzAn97} to perform binary search on all interpoint distances of $P$. In this paper, we gave two algorithms for the $L_2$ unweighted case and their time complexities are $O(\lfloor \lambda \rfloor \cdot n \log n)$ and $O(n^{5/4} \log^{7/4} n)$, respectively; we also gave an algorithm of $O(n^{5/4} \log^{5/2} n)$ time for the $L_2$ weighted case. In addition, we solve the $L_1$ RSP problem in $O(n \log^3 n)$ time for both the unweighted and weighted cases.

Since the original reporting of our results,\footnote{Our algorithms for the $L_2$ unweighted case were included in~\cite{ref:WangRe21}; our results for the $L_2$ weighted case and the $L_1$ problem have been presented in the 29th Fall Workshop on Computational Geometry (FWCG 2021) and has also been accepted in~\cite{ref:WangRe22}. Note that the second algorithm for the $L_2$ unweighted case runs in $O(n^{5/4} \log^{2} n)$ time in~\cite{ref:WangRe21}; in this full version, we slightly improve the time to $O(n^{5/4} \log^{7/4} n)$ by changing the threshold for defining large cells from $n^{3/4}$ to $(n/\log n)^{3/4}$ in Section~\ref{sec:second}.} some exciting progress has been made by Katz and Sharir~\cite{ref:KatzEf21arXiv}, who proposed randomized algorithms of $O(n^{6/5+\epsilon})$ expected time for the $L_2$ RSP problem for both the unweighted and weighted cases, for any arbitrarily small $\epsilon>0$.\footnote{It is not explicitly stated in~\cite{ref:KatzEf21arXiv} that the algorithm is randomized. A key subroutine used in the algorithm is Theorem~1, which is originally from \cite{ref:AvrahamTh15} and is a randomized algorithm (see Section~4 in \cite{ref:AvrahamTh15}).} Note that all our results are deterministic.

Note that reverse/inverse shortest path problems have been studied in the literature under various problem settings. Roughly speaking, the problems are to modify the graph (e.g., modify some edge weights) so that certain desired constraints related to shortest paths in the graph can be satisfied, e.g.,~\cite{ref:BurtonOn92,ref:ZhangCo03}. Our reverse shortest path problem in unit-disk graphs may find applications in scenarios like the following. Consider $G_r(P)$ as an $L_2$ unit-disk intersection graph representing a wireless sensor network in which each disk represents a sensor and two sensors can communicate with each other (e.g., directly transmit a message) if there is an edge connecting them in $G_r(P)$. The disk radius is proportional to the energy of the sensor. For two specific sensors $s$ and $t$, suppose we want to know the minimum energy for all sensors so that $s$ and $t$ can transmit messages to each other within $\lambda$ steps for a given value $\lambda$. It is easy to see that this is equivalent to our $L_2$ RSP problem in the unweighted case. If the latency of transmitting a message between two neighboring sensors is proportional to their Euclidean distance and we want to know the minimum energy for all sensors so that the total latency of transmitting messages between $s$ and $t$ is no more than a target value $\lambda$, then the problem becomes the weighted case.

In addition to the shortest path problem, many other problems of unit-disk graphs have also been studied, i.e. clique~\cite{ref:ClarkUn90}, independent set~\cite{ref:MatsuiAp98}, distance oracle~\cite{ref:ChanAp18, ref:GaoWe05}, diameter~\cite{ref:ChanAl16, ref:ChanAp18, ref:GaoWe05}, etc. Comparing to general graphs, many problems can be solved efficiently in unit-disk graphs by exploiting their underlying geometric structures, although there are still problems that are NP-hard for unit-disk graphs and other geometric intersection graphs, e.g.,~\cite{ref:deBergAf18, ref:ClarkUn90}.

\subsection{Our approach}
\label{subsec:OurApproach}

As the length of any path in $G_r(P)$ is an integer in the unweighted case, the length of a path of $G_r(P)$ is at most $\lambda$ if and only if the length of the path is at most $\lfloor \lambda\rfloor$; therefore, we can replace $\lambda$ in the unweighted problem by $\lfloor \lambda\rfloor$.
In the following, we simply assume that $\lambda$ is an integer in the
unweighted case.
Recall that our goal is to compute $r^*$, which must be equal to the distance of two
points in $P$ in both the unweighted and weighted cases. Given a value
$r$, {\em the decision problem} is to decide whether $r\geq r^*$. It
is not difficult to see that $r\geq r^*$ if and only if the distance
of $s$ and $t$ in $G_r(P)$ is at most $\lambda$. Therefore, the
decision problem can be solved efficiently by using the
shortest path algorithm for the corresponding case~\cite{ref:CabelloSh15,ref:ChanAl16}. More specifically, with $O(n\log n)$-time
preprocessing (to sort the points of $P$), given any $r$, whether
$r\geq r^*$ can be decided in $O(n)$ time for the $L_2$ unweighted unit-disk graphs by
the algorithm of Chan and Skrepetos~\cite{ref:ChanAl16}.
For the $L_2$ weighted case, the decision problem can be solved
in $O(n \log^2 n)$ time by Wang and Xue's shortest path algorithm~\cite{ref:WangNe20}. As in the
$L_2$ unweighted case, the decision problem in the $L_1$ unweighted case can be solved in $O(n\log n)$
time by applying the SSSP algorithms for the $L_2$ unweighted case~\cite{ref:CabelloSh15,ref:ChanAl16,ref:WangRe21} (or $O(n)$ time after $O(n\log n)$ time preprocessing for
sorting the points of $P$~\cite{ref:ChanAl16}). The decision problem in the $L_1$ weighted case can be solved in $O(n \log n)$ time as well~\cite{ref:WangAn21}.

Since $r^*$ must be equal to the distance of two points of $P$, we can find $r^*$ by doing binary search on the set of pairwise distances of all points of $P$. Given any $1 \leq k \leq \binom{n}{2}$, the distance selection algorithm of Katz and Sharir~\cite{ref:KatzAn97} can compute the $k$-th smallest $L_2$ distance among all pairs of points of $P$ in $O(n^{4/3}\log^2 n)$ time. Using this algorithm, the binary search can find $r^*$ in $O(n^{4/3}\log^3 n)$ time for both the $L_2$ unweighted and weighted cases. This is the algorithm mentioned in~\cite{ref:CabelloSh15}.

Our algorithms for the $L_2$ RSP problem are based on parametric search~\cite{ref:ColeSl87,ref:MegiddoAp83}, by parameterizing the decision algorithm of Chan and Skrepetos~\cite{ref:ChanAl16} (which we refer to as the CS algorithm) in the unweighted case, and parameterizing the decision algorithm of Wang and Xue~\cite{ref:WangNe20} (which we refer to as the WX algorithm) in the weighted case. For the $L_1$ RSP problem, we use an approach similar to the $L_2$ distance selection algorithm in~\cite{ref:KatzAn97}. 
Below is an overview on our algorithms.

\paragraph{The $L_2$ unweighted case.}
The CS algorithm first builds a grid in the plane and then runs the breadth-first-search (BFS) algorithm with the help of the grid; in the $i$-th step of the BFS, the algorithm finds the set of points of $P$ whose distances from $s$ in $G_r(P)$ are equal to $i$. Although we do not know $r^*$, we run the CS algorithm on a parameter $r$ in an interval $(r_1,r_2]$ such that each step of the algorithm behaves the same as the CS algorithm running on $r^*$.
The algorithm terminates after $t$ is reached, which will happen within $\lambda$ steps. In each step, we use the CS algorithm to compare $r^*$ with certain {\em critical values}, and the interval $(r_1,r_2]$ will be shrunk based on the results of these comparisons. Once the algorithm terminates, $r^*$ is equal to $r_2$ of the current interval $(r_1,r_2]$. With the linear-time decision algorithm (i.e., the CS algorithm~\cite{ref:ChanAl16}), each step runs in $O(n\log n)$ time. The total time of the algorithm is $O(\lambda\cdot n\log n)$.

The above algorithm is only interesting when $\lambda$ is relatively small. In the worst case, however, $\lambda$ can be $\Theta(n)$, which would make the running time become $O(n^2\log n)$. Next, by combining the strategies of the parametric search and the $L_2$ distance selection algorithm~\cite{ref:KatzAn97}, we derive a better algorithm. The main idea is to partition the cells of the grid in the CS algorithm into two types: {\em large cells}, which contain at least $(n / \log n)^{3/4}$ points of $P$ each, and {\em small cells} otherwise. For small cells, we process them using the above binary search algorithm with the $L_2$ distance selection algorithm~\cite{ref:KatzAn97}; for large cells, we process them using the above parametric search techniques. This works out due to the following observation. On the one hand, the number of large cells is relatively small (at most $O(n^{1/4} \log^{3/4} n)$) and thus the number of steps using the parametric search is also small. On the other hand, each small cell contains relatively few points of $P$ (at most $O((n / \log n)^{3/4})$) and thus the total time we spend on the $L_2$ distance selection algorithm is not big. The threshold value $(n / \log n)^{3/4}$ is carefully chosen so that the total time for processing the two types of cells is minimized. In addition, instead of applying the $L_2$ distance selection algorithm~\cite{ref:KatzAn97} directly, we find that it suffices to use only a subroutine of that algorithm, which not only simplifies the algorithm but also reduces the total time by a logarithmic factor. All these efforts lead to an $O(n^{{5}/{4}} \log^{7/4} n)$ time algorithm to compute $r^*$.

\paragraph{The $L_2$ weighted case.}
Our algorithm for the $L_2$ weighted case also follows the parametric search scheme, by parameterizing the WX algorithm~\cite{ref:WangNe20} instead.
Like the unweighted case, we run the decision algorithm (i.e., the WX algorithm) with a parameter $r\in (r_1,
r_2]$ by simulating the decision algorithm on the unknown $r^*$. At each
step of the algorithm, we call the decision algorithm on certain critical values $r$ to compare $r$ and $r^*$, and the algorithm will proceed accordingly based on the result of the comparison. The
interval $(r_1,r_2]$ will also be shrunk after these comparisons but
is guaranteed to contain $r^*$ throughout the algorithm.
The algorithm terminates once the point $t$ is reached, at which moment we can prove that $r^*$ is equal to $r_2$ of the current interval $(r_1,r_2]$. The parametric search algorithm runs in $\Omega(n^2)$ time because $t$ may be reached after $\Theta(n)$ steps.
To further reduce the time, similarly to the $L_2$ unweighted case, we combine the strategies of the parametric search and the $L_2$ distance selection techniques~\cite{ref:KatzAn97}. The cells of the grid built in the algorithm are partitioned into large and small cells, but with a different threshold of $n^{3/4}\log^{3/2}n$.
With this approach, the runtime of the algorithm can be bounded by $O(n^{5/4} \log^{5/2} n)$.

\paragraph{The $L_1$ cases.} We use an approach similar to the $L_2$ distance
selection algorithm in~\cite{ref:KatzAn97}. Let $\Pi$ denote the set of all pairwise distances of all points of $P$.
In light of the observation that $r^*$ is in $\Pi$, each iteration
of our algorithm computes an interval $(a_j,b_j]$ (initially,
$a_0=-\infty$ and $b_0=\infty$) such that $r^*\in (a_j,b_j]$ and the
number of values of $\Pi$ in $(a_j,b_j]$ is a constant fraction of the
number of values of $\Pi$ in $(a_{j-1},b_{j-1}]$. In this way, $r^*$
can be found within $O(\log n)$ iterations. Each iteration will call
the corresponding decision algorithm~\cite{ref:ChanAl16, ref:WangAn21} to perform binary search on certain
values.
The total time of the algorithm for both the unweighted and weighted cases is $O(n\log^3 n)$.

A by-product of our technique is an $O(n\log^3 n)$ time
algorithm that can compute the $k$-th smallest $L_1$
distance among all pairs of points of $P$, for
any given $k$ with $1\leq k\leq \binom{n}{2}$. As mentioned before,
the $L_2$ version of the problem can be solved in $O(n^{4/3}\log^2 n)$
time~\cite{ref:KatzAn97}.

\paragraph{Outline.} The rest of the paper is organized as follows. Section~\ref{sec:pre} defines notation and reviews the CS algorithm. Our first algorithm for the $L_2$ unweighted case is presented in Section~\ref{sec:first} while the second one is described in Section~\ref{sec:second}. Section~\ref{sec:WeightedRSPalgorithm} solves the $L_2$ weighted RSP problem. Section~\ref{sec:L1RSP} gives our algorithm for the $L_1$ RSP problem for both the unweighted and weighted cases. Section~\ref{sec:con} concludes with remarks showing that our techniques can be readily extended to solve a more general ``single-source'' version of the RSP problem.

\section{Preliminaries}
\label{sec:pre}

Throughout the paper, we will use ``points of $P$'' and ``vertices of the graph $G_r(P)$'' interchangeably.
For any parameter $r$, let $d_r(p,q)$ denote the distance of two vertices $p$ and $q$ in $G_r(P)$. It is easy to see that $d_r(p,q)\leq d_{r'}(p,q)$ if $r\geq r'$.

For any two points $p$ and $q$ in the plane, let $\lVert p-q \rVert$ denote their Euclidean distance. For any subset $P'$ of $P$ and any region $R$ in the plane, we use $P'(R)$ or $P'\cap R$ to refer to the subset of points $P'$ contained in $R$. For any point $p$, let $x(p)$ and $y(p)$ denote its $x$- and $y$-coordinates, respectively.

We next review the CS algorithm~\cite{ref:ChanAl16}, which will help understand our RSP algorithms given later. Suppose we have a sorted list of $P$ by $x$-coordinate and another sorted list of $P$ by $y$-coordinate. Given a parameter $r$ and a source point $s\in P$, the CS algorithm can compute in $O(n)$ time the distances from $s$ to all other points of $P$ in $G_r(P)$.

The first step is to compute a grid $\Psi_r(P)$ of square cells whose side lengths are $r/\sqrt{2}$. A cell $C'$ of $\Psi_r(P)$ is a {\em neighbor} of another cell $C$ if the minimum distance between a point of $C$ and a point of $C'$ is at most $r$. Note that the number of neighbors of each cell of $\Psi_r(P)$ is $O(1)$ (e.g., see Fig.~\ref{fig:grid}) and the distance between any two points in each cell is at most $r$.

\begin{figure}[t]
\begin{minipage}[t]{\textwidth}
\begin{center}
\includegraphics[height=2.0in]{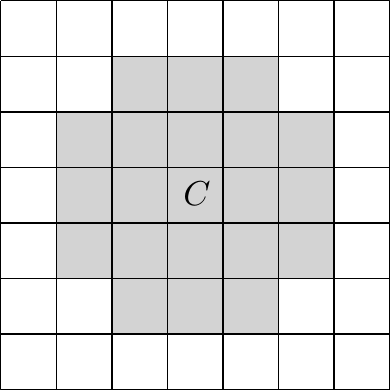}
\caption{\footnotesize The grey cells are all neighbor cells of $C$.}
\label{fig:grid}
\end{center}
\end{minipage}
\end{figure}

Next, starting from the point $s$, the algorithm runs BFS in $G_r(P)$ with the help of the grid $\Psi_r(P)$. Define $S_i$ as the subset of points of $P$ whose distances in $G_r(P)$ from $s$ are equal to $i$.
Initially, $S_0=\{s\}$. Given $S_{i-1}$, the $i$-th step of the BFS is to compute $S_{i}$ by using $S_{i-1}$ and the grid $\Psi_r(P)$, as follows. If a point $p$ is not in $\bigcup_{j=0}^{i-1}S_j$, we say that $p$ has not been {\em discovered} yet. For each cell $C$ that contains at least one point of $S_{i-1}$, we need to find points that are not discovered yet and at distances at most $r$ from the points of $S_{i-1}\cap C$ (i.e., the points of $S_{i-1}$ in $C$); clearly, these points are either in $C$ or in the neighbor cells of $C$. For points of $P(C)$, since every two points of $C$ are within distance $r$ from each other, we add all points of $P(C)$ that have not been discovered to $S_i$. For each neighbor cell $C'$ of $C$, we need to solve the following subproblem: find the points of $P(C')$ that are not discovered yet and within distance at most $r$ from the points of $S_{i-1}\cap C$. Since $C'$ and $C$ are separated by either a vertical line or a horizontal line, we essentially have the following subproblem.

\begin{subproblem}\label{subpro:10}
Given a set of $n_r$ red points below a horizontal line $\ell$ and a set of $n_b$ blue points above $\ell$, both sorted by $x$-coordinate, determine for each blue point whether there is a red point at distance at most $r$ from it.
\end{subproblem}

The subproblem can be solved in $O(n_r+n_b)$ time as follows. For each red point $p$, the circle of radius $r$ centered at $p$ has at most one arc above $\ell$ (we say that this arc is defined by $p$). Let $\Gamma$ be the set of these arcs defined by all red points.
Since all arcs of $\Gamma$ have the same radius and all red points are below $\ell$, every two arcs intersect at most once and the arcs above $\ell$ are $x$-monotone. Further, as all red points are sorted already by $x$-coordinate, the upper envelope of $\Gamma$, denoted by $\calU$, can be computed in $O(n_r)$ time by an algorithm similar in spirit to Graham's scan. Then, it suffices to determine whether each blue point is below $\calU$, which can be done in $O(n_r+n_b)$ time by a linear scan. More specifically, we can first sort the vertices of $\calU$ and all blue points. After that, for each blue point $p$, we know the arc of $\calU$ that spans $p$ (i.e., $x(p)$ is between the $x$-coordinates of the two endpoints of the arc), and thus we only need to check whether $p$ is below the arc. In summary, solving the subproblem involves three subroutines: (1) compute $\calU$; (2) sort all vertices of $\calU$ with all blue points; (3) for each blue point $p$, determine whether it is below the arc of $\calU$ that spans $p$.

The above computes the set $S_{i}$. Note that if $S_{i}=\emptyset$, then we can stop the algorithm because all points of $P$ that can be reached from $s$ in $G_r(P)$ have been computed. For the running time, notice that points of $P$ in each cell of the grid $\Psi_r(P)$ can be involved in at most two steps of the BFS. Further, since each grid cell has $O(1)$ neighbors, the total time of the BFS algorithm is $O(n)$.

In order to achieve $O(n)$ time for the overall algorithm, the grid $\Psi_r(P)$ must be implicitly constructed. The CS algorithm~\cite{ref:ChanAl16} does not provide any details about that. There are various ways to do so. Below we present our method, which will facilitate our algorithm in the next section.

The grid $\Psi_r(P)$ we are going to build is a rectangle that is partitioned into square cells of side lengths $r/\sqrt{2}$ by $O(n)$ horizontal and  vertical lines. These partition lines will be explicitly computed. Let $P'$ be the subset of points of $P$ located in $\Psi_r(P)$. $P'$ has the following property: for each $p\in P\setminus P'$, $p$ cannot be reached from $s$ in $G_r(P)$, i.e., the distances from $s$ to the points of $P\setminus P'$ in $G_r(P)$ are infinite. Let $\calC$ denote the set of cells of $\Psi_r(P)$ that contain at least one point of $P$. For each cell $C\in \calC$, let $N(C)$ denote the set of neighbors of $C$ in $\calC$.
The information computed in the following lemma suffices for implementing the above BFS algorithm in linear time.

\begin{lemma}\label{lem:grid}
Suppose we have a sorted list of $P$ by $x$-coordinate and another sorted list of $P$ by $y$-coordinate.
Both $P'$ and $\calC$, along with all vertical and horizontal partition lines of $\Psi_r(P)$, can be computed in $O(n)$ time. Further, with $O(n)$ time preprocessing, the following can be achieved:
\begin{enumerate}
  \item Given any point $p\in P'$, the cell of $\calC$ that contains $p$ can be obtained in $O(1)$ time.
  \item Given any cell $C\in \calC$, the neighbor set $N(C)$ can be obtained in $O(|N(C)|)$ time.
  \item Given any cell $C\in \calC$, the subset $P(C)$ of $P$ can be obtained in $O(|P(C)|)$ time.
\end{enumerate}
\end{lemma}
\begin{proof}
Let $P_1$ be the subset of $P$ to the right of $s$ including $s$. Let $s=p_1,p_2,\ldots,p_m$ be the list of $P_1$ sorted from left to right, with $m=|P_1|$. As the points of $P$ are given sorted, we can obtain the above sorted list in $O(n)$ time. During the algorithm, we will compute a subset $Q\subseteq P$. Initially, we set $Q=\emptyset$. After the algorithm finishes, we will have $P'=P\setminus Q$.

We find the smallest index $i\in [1,m-1]$ such that $x(p_{i+1})-x(p_i)>r$ (let $i=m$ if such index does not exist). It is easy to see for any point $p_j$ with $j\in [i+1,m]$, there is no path from $s$ to $p_j$ in $G_r(P)$. We add all points $p_{i+1},p_{i+2},\ldots, p_m$ to $Q$ and let $P_1'=\{p_1,\ldots,p_i\}$. Hence, $P_1'$ has the following property: $x(p_{j+1})-x(p_j)\leq r$ for any two adjacent points $p_j$ and $p_{j+1}$.
Next, we compute the vertical partition lines of $\Psi_r(P)$ to the right of $s$.
We first put a vertical line through $s$. Then, we keep adding a vertical line to the right with horizontal distance $r/\sqrt{2}$ from the previous vertical line until the current vertical line is to the right of $p_i$. Due to the above property of $P_1'$, the number of vertical lines thus produced is at most $2m$.

The above computes a set of vertical partition lines to the right of $s$ by considering the points of $P_1$ from left to right. Let $P_2=P\setminus P_1$; we also add $s$ to $P_2$. Symmetrically, we compute a set of vertical partition lines to the left of $s$ by considering the points of $P_2$ from right to left (also starting from $s$). Analogously, the algorithm will compute a subset $P_2'$ of $P_2$ and more points may be added to $Q$. Let $L_v$ be the set of all these vertical lines produced above for both $P_1$ and $P_2$. $L_v$ is the set of vertical partition lines of our grid $\Psi_r(P)$. Clearly, $|L_v|=O(n)$.

Similarly, by considering the points of $P$ in the list sorted by $y$-coordinate, we can compute a set $L_h$ of horizontal partition lines of $\Psi_r(P)$, with $|L_h|=O(n)$. Also, more points may be added to $Q$ in the process.

Let $\Psi_r(P)$ be the rectangle bounded by the rightmost and leftmost vertical lines of $L_v$ as well as the topmost and bottommost horizontal lines of $L_h$, along with the square cells inside and partitioned by the lines of $L_v\cup L_h$. Let $P'=P\setminus Q$. By our definition of $Q$, for each $p\in Q$, $p$ cannot be reached from $s$ in $G_r(P)$, and $P'$ is exactly the subset of points of $P$ located inside $\Psi_r(P)$.

For each cell $C$ of $\Psi_r(P)$, we define its {\em grid-coordinate} as $(i,j)$ if $C$ is in the $i$-th row  and $j$-th column of $\Psi_r(P)$; we say that $i$ is the {\em row-coordinate} and $j$ is the {\em column-coordinate}. For each cell, we consider its grid-coordinate as its ``ID''.

By scanning the points of $P'$ and the vertical lines of $L_v$ from left to right and then scanning $P'$ and the horizontal lines of $L_h$ from top to bottom, we can compute in $O(n)$ time for each point of $P'$ the (grid-coordinate of the) cell of $\Psi_r(P)$ that contains it (to resolve the boundary case, if a point $p$ is on a vertical edge shared by two cells, then we assume $p$ is contained in the right cell only, and if $p$ is on a horizontal edge shared by two cells, then we assume $p$ is contained in the top cell only). After that, given any point $p\in P'$, the  cell of $\Psi_r(P)$ that contains $p$ can be obtained in $O(1)$ time.

To compute the set $\calC$, we do the following. Initialize $\calC=\emptyset$. Then, for each point $p\in P'$, we add the cell that contains $p$ into $\calC$. Note that $\calC$ may be a multi-set. To remove the duplicates, we first sort all cells of $\calC$ by their grid-coordinates in lexicographical order (i.e., compare row-coordinates first and then column-coordinates). This sorting can be done in $O(n)$ time by radix sort~\cite{ref:CLRS09}, because both the row-coordinate and the column-coordinate of each cell are in the range $[1,O(n)]$. Now we can remove duplicates by simply scanning the sorted list of all cells, and the resulting set is $\calC$. Also, during the scanning process, we can obtain for each cell $C$ of $\calC$ the subset $P(C)$ of points of $P$ contained in $C$ (each occurrence of $C$ in the sorted list corresponds to a point of $P$ that is contained in $C$). All these can be done in $O(n)$ time. After that, given each cell $C$ of $\calC$, we can output $P(C)$ in $O(|P(C)|)$ time.

It remains to compute the neighbor set $N(C)$ for each cell $C\in \calC$. This can be done in $O(n)$ time by scanning the above sorted list of $\calC$ (after the duplicates are removed). Indeed, notice that scanning the sorted list is equivalent to scanning the non-empty cells of $\Psi_r(P)$ row by row and from left to right in each row. Recall that  the cells of $N(C)$ are in at most five rows of the grid (e.g., see Fig.~\ref{fig:grid}): the row containing $C$, two rows above it, and two rows below it; each such row contains at most fives cells of $N(C)$. Based on this observation, we scan the cells in the sorted list of $\calC$. For each cell $C$ under consideration during the scan, suppose its grid-coordinate is $(i,j)$. During the scan, we maintain a cell $(i',j')\in \calC$ in each row $i'$ for $i'\in \{i-2,i-1,i,i+1,i+2\}$ such that $j'$ is closest to $j$, i.e., $|j'-j|$ is minimized (e.g., for $i'=i$, we have $j'=j$). Using these cells, we can find $N(C)$ in $O(1)$ time (indeed, for each row $i'\in\{i-2,i-1,i,i+1,i+2\}$, the cells of $N(C)$ contained in row $i'$ are within five cells of $(i',j')$ in the sorted list of $\calC$). The scan can be implemented in $O(n)$ time. After that, $N(C)$ for all cells $C\in\calC$ are computed.
This proves the lemma.
\end{proof}

To make the description concise, in the following, whenever we say ``compute the grid $\Psi_r(P)$'' we mean ``compute the grid information of Lemma~\ref{lem:grid}''; similarly, by ``using the grid $\Psi_r(P)$'', we mean ``using the grid information computed by Lemma~\ref{lem:grid}''.

\section{The $L_2$ unweighted case -- the first algorithm}
\label{sec:first}

In this section, we present our $O(\lambda\cdot n\log n)$ time algorithm for the unweighted RSP problem. Given $\lambda$ and $s,t\in P$, our goal is to compute $r^*$, the optimal radius of the disks.

As discussed in Section~\ref{subsec:OurApproach}, our algorithm uses parametric search~\cite{ref:ColeSl87,ref:MegiddoAp83}. But different than the traditional parametric search where parallel algorithms are used, our decision algorithm (i.e., the CS algorithm for the shortest path problem~\cite{ref:ChanAl16}) is inherently sequential. We will run the CS algorithm with a parameter $r$ in an interval $(r_1,r_2]$ by simulating the algorithm on the unknown $r^*$; at each step of the algorithm, the decision algorithm will be invoked on certain {\em critical values} $r$ to compare $r$ and $r^*$, and the algorithm will proceed accordingly based on the results of the comparisons.
The interval $(r_1,r_2]$ always contains $r^*$ and will keep shrinking during the algorithm (note that ``shrinking'' includes the case that the interval does not change). Initially, we set $r_1=0$ and $r_2=\infty$. Clearly, $(r_1,r_2]$ contains $r^*$.

Recall that the CS algorithm has two major steps: build the grid and then run BFS with the help of the grid. Correspondingly, our algorithm also first builds a grid and then runs BFS accordingly using the grid.

\subsection{Building the grid}
\label{sec:grid}

The first step is to build a grid $\Psi(P)$. Our goal is to shrink $(r_1,r_2]$ so that it contains $r^*$ and if $r^*\neq r_2$ (and thus $r^*\in (r_1,r_2)$), then for any $r\in (r_1,r_2)$, $\Psi_r(P)$ has the same {\em combinatorial structure} as $\Psi_{r^*}(P)$, i.e.,
both grids have the same number of columns and the same number of rows, and
a point of $P$ is in the cell of the $i$-th row and $j$-th column of $\Psi_{r^*}(P)$ if and only if it is also in the cell of the $i$-th row and $j$-th column of $\Psi_r(P)$. To this end, we have the following lemma.

\begin{lemma}
    \label{lem:20}
    An interval $(r_1,r_2]$ containing $r^*$ can be computed in $O(n\log n)$ time so that if $r^*\neq r_2$, then  for any $r\in (r_1,r_2)$, the grid $\Psi_r(P)$ has the same combinatorial structure as $\Psi_{r^*}(P)$.
\end{lemma}
\begin{proof}
Let $P_1$ be the subset of $P$ to the right of $s$ including $s$.
Let $s=p_1,p_2,\ldots,p_m$ be the list of $P_1$ sorted from left to right, with $m=|P_1|$.
Recall from the proof of Lemma~\ref{lem:grid} that $\Psi_{r^*}(P)$ has at most $2m$ vertical partition lines to the right of $s$, and there is a vertical partition line through $s$.

We first implicitly form a sorted matrix and then apply the sorted-matrix searching techniques of Frederickson and Johnson~\cite{ref:FredericksonGe84, ref:FredericksonFi83} to shrink $(r_1,r_2]$. Specifically, we define an $m\times 2m$ matrix $M$ with
$$M[i,j] = \sqrt{2} \cdot \frac{x(p_{i}) - x(p_1)}{j}$$ for all $1 \leq i \leq m$ and $1 \leq j \leq 2m$. It can be verified that $M[i,j] \geq  M[i,j+1]$ and $M[i+1,j] \geq  M[i,j]$ hold. Thus, $M$ is a sorted matrix.
Using the sorted-matrix searching techniques~\cite{ref:FredericksonGe84, ref:FredericksonFi83} with the CS algorithm as the decision algorithm, we can compute in $O(n\log n)$ time the largest value $r_1'$ of $M$ with $r_1'< r^*$ and the smallest value $r_2'$ of $M$ with $r^*\leq r_2'$. By definition, $(r_1',r_2']$ contains $r^*$ and $(r_1',r_2')$ does not contain any value of $M$. We update $r_1=\max\{r_1',r_1\}$ and $r_2=\min\{r_2',r_2\}$. Thus, the new interval $(r_1,r_2]$ shrinks but still contains $r^*$. As $(r_1,r_2)\subseteq(r_1',r_2')$, $(r_1,r_2)$ does not contain any value of $M$.

According to our algorithm of Lemma~\ref{lem:grid}, there is always a vertical partition line through $s$ in $\Psi_{r}(P)$ for any $r$. Let $\Psi^1_{r}(P)$ and $\Psi^2_{r}(P)$ refer to the half grids of $\Psi_{r}(P)$ to the right and left of $s$, respectively; assume that both half grids contain the vertical partition line through $s$. We claim that if $r^*\neq r_2$, then the following hold for any $r\in (r_1,r_2)$: (1) a point of $P_1$ is in the $j$-th column of $\Psi^1_{r^*}(P)$ if and only if it is also in the $j$-th column of $\Psi^1_r$(P); (2) the number of columns of $\Psi^1_{r^*}(P)$ is equal to the number of columns of $\Psi^1_{r}(P)$. We prove the claim below.

\begin{figure}[t]
    \centering
    \begin{minipage}{0.48\textwidth}
        \centering
        \includegraphics[height=2.0in]{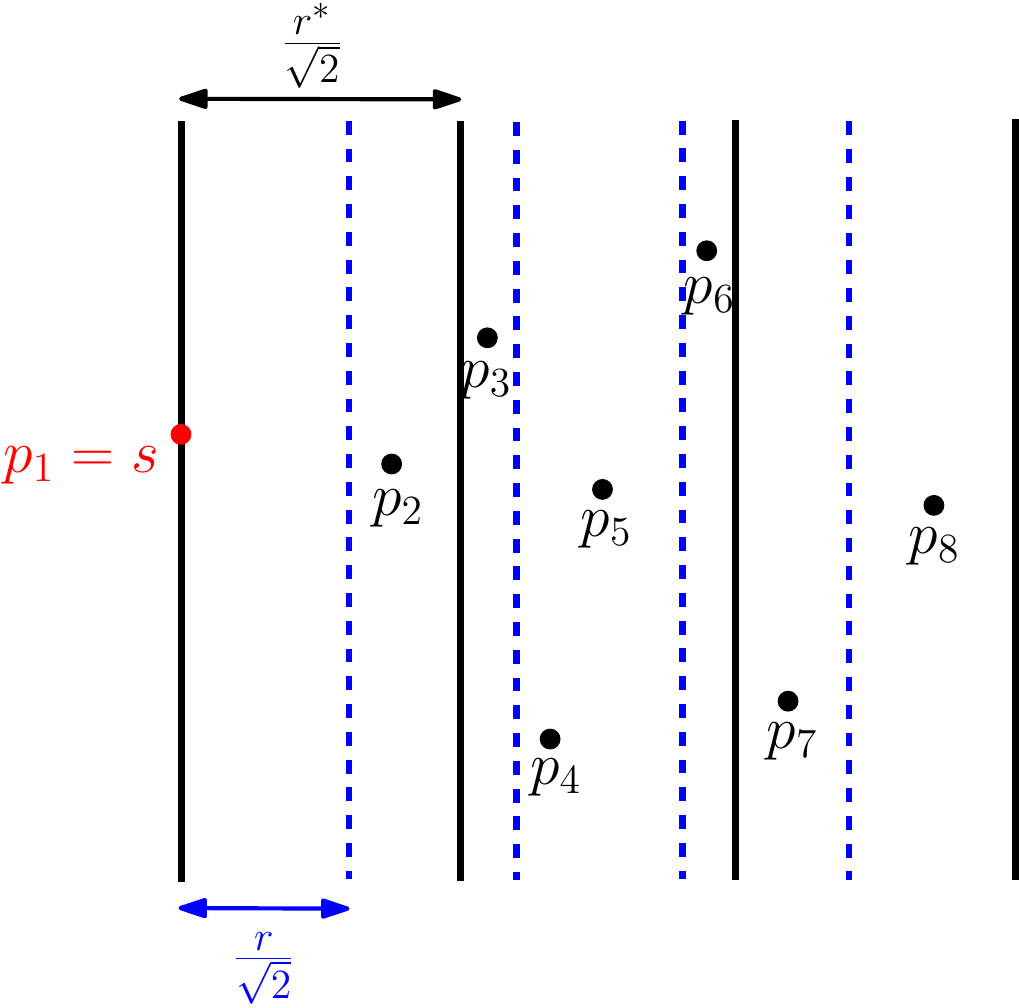}
        \caption{\footnotesize The point $p_7$ is in the 3rd column of $\Psi_{r^*}^1(P)$ while it is in the 4th column of $\Psi_{r}^1(P)$.}
        \label{fig:BuildColumns}
    \end{minipage}
    \hspace{0.05in}
    \begin{minipage}{0.49\textwidth}
        \centering
        \includegraphics[height=2.0in]{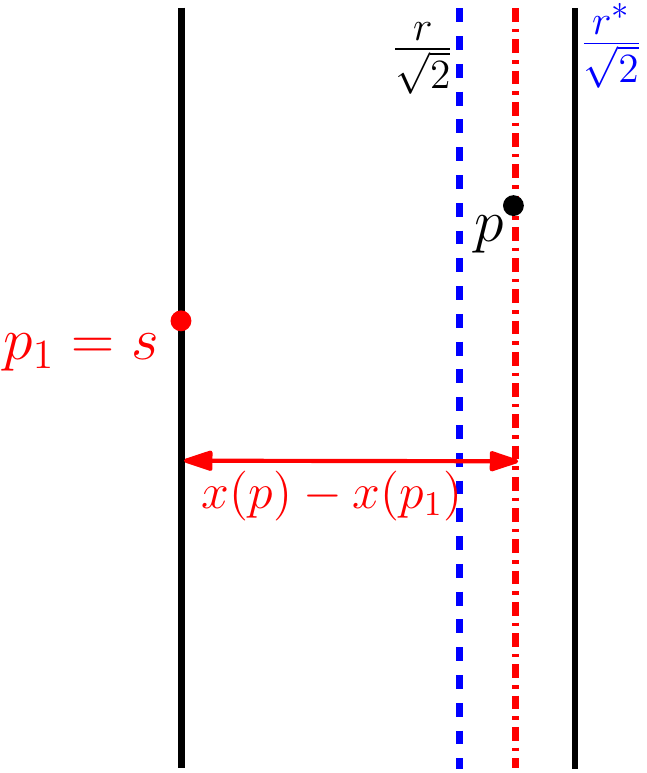}
        \caption{\footnotesize The rightmost line is $\ell$ when $r'=r^*$. When $r'$ decreases from $r^*$ to $r$, $\ell$ will move leftwards and cross $p$.}
        \label{fig:GridCriticalValue}
    \end{minipage}
\end{figure}

Suppose $r^*\neq r_2$. Then, $r^*\in (r_1,r_2)$.
Assume to the contrary that a point $p$ of $P_1$ is in the $j$-th column of $\Psi^1_{r^*}(P)$ for some $j\in [1,2m]$, but $p$ is not in the $j$-th column of $\Psi^1_r(P)$. Then, $p$ is either to the left or to the right of the $j$-th column of $\Psi^1_r(P)$. Without loss of generality, we assume that $p$ is to the right of the $j$-th column of $\Psi^1_r(P)$ (e.g., see Fig.~\ref{fig:BuildColumns}).
This implies that $r<r^*$. Further, if we decrease a value $r'$ gradually from $r^*$ to $r$, then the line $\ell$ will move monotonically leftwards and cross $p$ at some moment, where $\ell$ is the $(j+1)$-th vertical partition line of $\Psi^1_{r'}(P)$ (i.e., $\ell$ is the vertical bounding line of the $j$-th column of $\Psi^1_{r'}(P)$); e.g., see Fig.~\ref{fig:GridCriticalValue}. This further implies that $r/\sqrt{2}< (x(p)-x(p_1))/j < r^*/\sqrt{2}$, and thus, $r< \sqrt{2}\cdot (x(p)-x(p_1))/j < r^*$. On the other hand, since both $r$ and $r^*$ are in $(r_1,r_2)$, we obtain that $\sqrt{2}\cdot (x(p)-x(p_1))/j\in (r_1,r_2)$. Because the interval $(r_1,r_2)$ does not contain any values of $M$, we obtain contradiction as $\sqrt{2}\cdot (x(p)-x(p_1))/j$ is a value of $M$.

    Assume to the contrary that a point $p$ of $P_1$ is in the $j$-th column of $\Psi^1_{r}(P)$ for some $j\in [1,2m]$, but $p$ is not in the $j$-th column of $\Psi^1_{r^*}(P)$. Then, by similar analysis as above, we can obtain contradiction as well. This proves the first part of the claim.

The second part of the claim can actually be implied by the first part. Indeed, assume to the contrary that the number of columns of $\Psi^1_{r^*}(P)$, denoted by $m_{r^*}$, is not equal to the number of columns of $\Psi^1_{r}(P)$, denoted by $m_r$. Without loss of generality, we assume $m_{r^*}<m_r$.
By the algorithm of Lemma~\ref{lem:grid}, $P_1$ has a point $p$ in the last column of $\Psi^1_{r}(P)$, which is the $m_r$-th column. In light of the first part of the claim, $p$ is also in the $m_r$-th column of $\Psi^1_{r^*}(P)$. But this contradicts with that $\Psi^1_{r^*}(P)$ has only $m_{r^*}<m_r$ columns.

The claim is thus proved.

\medskip

The above processes the subset $P_1$ of $P$. Let $P_2=P\setminus P_1$; we add $s$ to $P_2$ as well.
Next, we use the same algorithm as above to process the points of $P_2$ and obtain a smaller interval $(r_1,r_2]$ containing $r^*$ such that if $r^*\neq r_2$, then the following hold for any $r\in (r_1,r_2)$: (1) a point of $P_2$ is in the $j$-th column of $\Psi^2_{r^*}(P)$ if and only if it is also in the $j$-th column of $\Psi^2_r$(P); (2) the number of columns of $\Psi^2_{r^*}(P)$ is equal to the number of columns of $\Psi^2_{r}(P)$. Combining the previous claim for $P_1$, we obtain that the interval $(r_1,r_2]$ contains $r^*$ and if $r^*\neq r_2$, then the following hold for any $r\in (r_1,r_2)$: (1) a point of $P$ is in the $j$-th column of $\Psi_{r^*}(P)$ if and only if it is also in the $j$-th column of $\Psi_r(P)$; (2) the number of columns of $\Psi_{r^*}(P)$ is equal to the number of columns of $\Psi_{r}(P)$.

The above processes the points of $P$ horizontally. We then process them in a vertical manner analogously and further shrink the interval $(r_1,r_2]$ such that it still contains $r^*$ and if $r^*\neq r_2$, then the following hold for any $r\in (r_1,r_2)$: (1) a point of $P$ is in the $i$-th row of $\Psi_{r^*}(P)$ if and only if it is also in the $i$-th row of $\Psi_r(P)$; (2) the number of rows of $\Psi_{r^*}(P)$ is equal to the number of rows of $\Psi_{r}(P)$.
As the interval $(r_1,r_2]$ is shrunk after processing $P$ vertically, we obtain that if $r^*\neq r_2$, then $\Psi_r(P)$ has the same combinatorial structure as $\Psi_{r^*}(P)$ for any $r\in (r_1,r_2)$. This proves the lemma.
\end{proof}

%

Let $(r_1,r_2]$ be the interval computed by Lemma~\ref{lem:20}.
We pick any value $r$ in $(r_1,r_2)$ and compute the grid $\Psi_r(P)$, i.e., compute the grid information of $\Psi_r(P)$ by Lemma~\ref{lem:grid}. By Lemma~\ref{lem:20}, these information is the same as that of $\Psi_{r^*}(P)$ if $r^*\neq r_2$. Below we will use $\Psi(P)$ to refer to the grid information computed above.

\subsection{Running BFS}
\label{subsec:bfs}

For a fixed parameter $r$, we use $S_i(r)$ to denote the set of points of $P$ whose distances from $s$ is equal to $i$ in $G_r(P)$, which is computed in the $i$-th step of the BFS algorithm if we run the CS algorithm with respect to $r$. Initially, we have $S_0(r)=\{s\}$. In the following, using the interval $(r_1,r_2]$ obtained in Lemma~\ref{lem:20}, we run the BFS algorithm as in the CS algorithm with a parameter $r\in (r_1,r_2)$, by simulating the algorithm for $r^*$. The algorithm maintains an invariant that the $i$-th step computes a subset $S_i\subseteq P$ and shrinks $(r_1,r_2]$ so that it contains $r^*$ and if $r^*\neq r_2$ (and thus $r^*\in (r_1,r_2)$), then $S_i=S_i(r)=S_i(r^*)$ for any $r\in (r_1,r_2)$. Initially, we set $S_0=\{s\}$ and thus the invariant holds as $S_0(r)=\{s\}$ for any $r$. As will be seen later, the algorithm stops within $\lambda$ steps and each step takes $O(n\log n)$ time.

Consider the $i$-th step. Assume that we have $S_{i-1}$ and $(r_1,r_2]$, and the invariant holds, i.e., $(r_1,r_2]$ contains $r^*$ and if $r^*\neq r_2$, then $S_{i-1}=S_{i-1}(r)=S_{i-1}(r^*)$ for any $r\in (r_1,r_2)$.
Using the grid $\Psi(P)$, we obtain the grid cells containing the points of $S_{i-1}$. For each such cell $C$, for points of $P$ in $C$, we have the following observation.

\begin{lemma}
Suppose $r^*\neq r_2$. Then, for each point $p\in P(C)$ that has not been discovered by the algorithm yet, i.e., $p\not\in \bigcup_{j=1}^{i-1} S_j$, $p$ is in $S_i(r)$ for any $r\in (r_1,r_2)$.
\end{lemma}
\begin{proof}
Let $q$ be a point of $S_{i-1}$ in $C$. By our algorithm invariant, $(r_1,r_2]$ contains $r^*$. Since $r^*\neq r_2$, $r^*\in (r_1,r_2)$. Let $r$ be any value of $(r_1,r_2)$. In light of Lemma~\ref{lem:20}, both $p$ and $q$ are in the same cell of $\Psi_{r}(P)$, and thus $\lVert p-q \rVert \leq r$. By our algorithm invariant, $S_j=S_j(r)$ for all $0\leq j\leq i-1$. Since $p\not\in \bigcup_{j=1}^{i-1} S_j$, we have $p\not\in \bigcup_{j=1}^{i-1} S_j(r)$. Because $q\in S_{i-1}(r)$ and $\lVert p-q \rVert \leq r$, we obtain that $p\in S_i(r)$.
\end{proof}

Due to the preceding lemma, we add to $S_i$ the points of $P(C)$ that have not been discovered yet. Next, for each neighbor $C'$ of $C$, we need to solve Subproblem~\ref{subpro:10}; we use $\calI$ to denote the set of all instances of this subproblem in the $i$-th step of the BFS.
Consider one such instance. Recall that solving it for a fixed $r$ involves three subroutines. First, compute the upper envelope $\calU$ of the arcs of $\Gamma$ above $\ell$ of all red points. Second, sort all vertices of $\calU$ with all blue points. Third, for each blue point $p$, determine whether it is below the arc of $\calU$ that spans $p$. To solve our problem, we parameterize each subroutine with a parameter $r$ so that the behavior of the algorithm is consistent with that for $r=r^*$ if $r^*\neq r_2$.

\subsubsection{Computing the upper envelope}
\label{subsubsec:ComputingTheUpperEnvelope}

We use $\Gamma(r)$ to denote the set of arcs above $\ell$ defined by the red points with respect to the  radius $r$; similarly, define $\calU(r)$ as the upper envelope of $\Gamma(r)$.

\begin{figure}[t]
    \centering
    \subfloat[The upper envelope is comprised of three arcs centered at $p_1$, $p_2$ and $p_3$.]{
    \label{fig:DetermineUpperEnvelope-a}
    \begin{minipage}[t]{0.3\linewidth}
    \centering
    \includegraphics[scale=0.5]{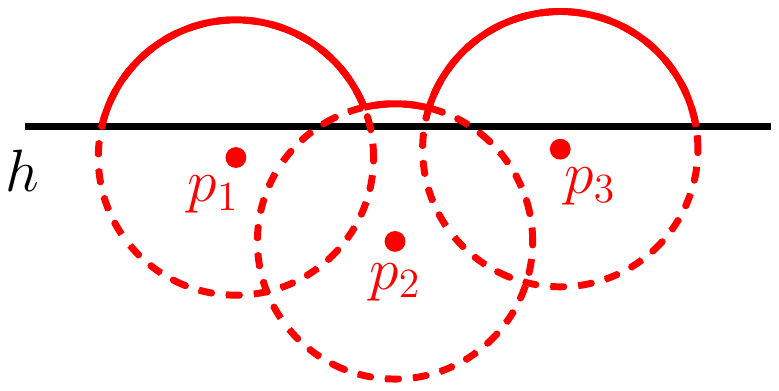}
    \end{minipage}%
    }
    \hspace{0.02\linewidth}
    \subfloat[The moment when the three arcs have a common intersection, which is a vertex of the upper envelope.]{
    \label{fig:DetermineUpperEnvelope-b}
    \begin{minipage}[t]{0.3\linewidth}
    \centering
    \includegraphics[scale=0.5]{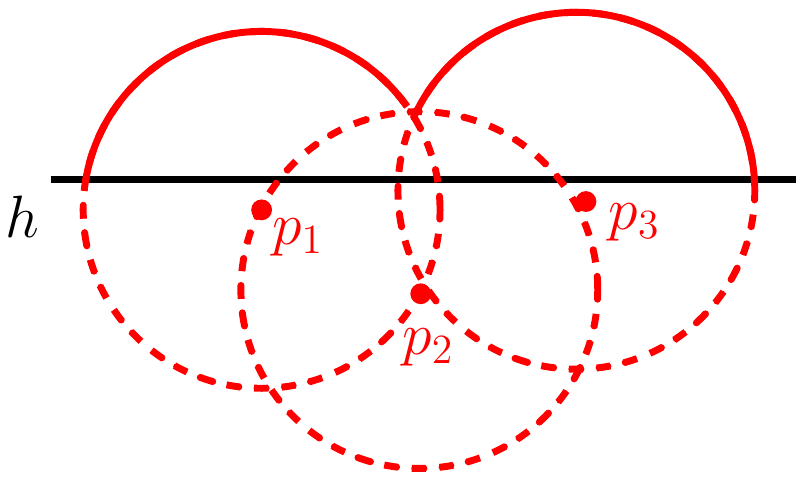}
    \end{minipage}%
    }%
    \hspace{0.02\linewidth}
    \subfloat[The middle arc centered at $p_2$ disappears from the upper envelope.]{
    \label{fig:DetermineUpperEnvelope-c}
    \begin{minipage}[t]{0.3\linewidth}
    \centering
    \includegraphics[scale=0.5]{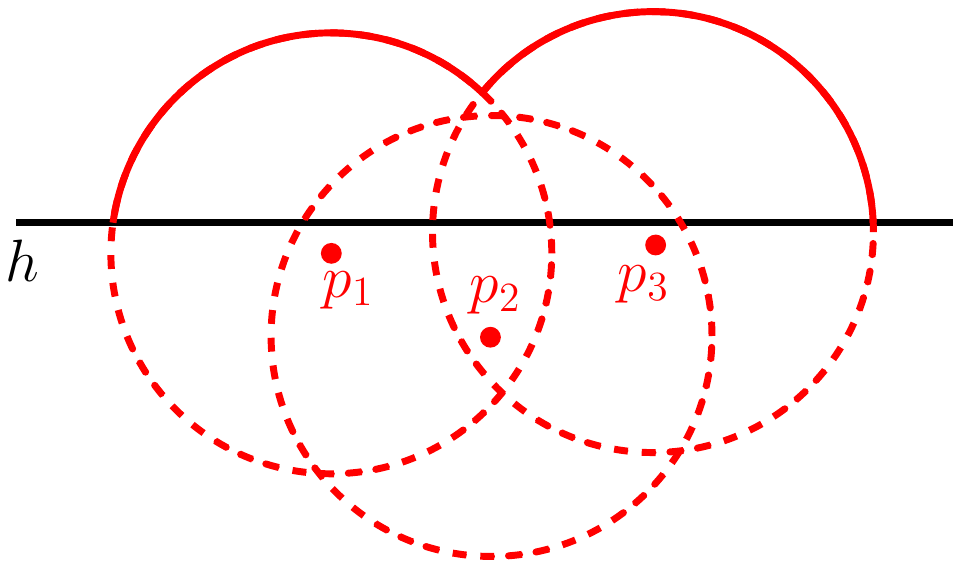}
    \end{minipage}
    }%
    \caption{\footnotesize The change of the combinatorial structure of the upper envelope $\calU(r)$ (the red solid arcs) as $r$ increases.}
    \label{fig:DetermineUpperEnvelope}
\end{figure}



The goal of the first subroutine is to shrink the interval $(r_1,r_2]$ such that it contains $r^*$ and if $r^*\neq r_2$, then $\calU(r^*)$ has the same combinatorial structure as $\calU(r)$ for any $r\in (r_1,r_2)$, i.e., the set of red points that define the arcs on $\calU(r)$ is exactly the set of red points that define the arcs on $\calU(r^*)$ with the same order. Note that the order of the arcs on $\calU(r)$ is consistent with the $x$-coordinate order of the red points defining these arcs~\cite{ref:ChanAl16}.

To this end, we have the following observation. Consider $\calU(r)$ for an arbitrary $r$. If $r$ changes, the combinatorial structure of $\calU(r)$ does not change until one arc (e.g., defined by a red point $p_2$) disappears from $\calU(r)$ (e.g., see Fig.~\ref{fig:DetermineUpperEnvelope}). Let $p_1$ and $p_3$ be the red points defining neighboring left and right arcs of the arc defined by $p_2$ on $\calU(r)$, respectively. Then, at the moment  when $p_2$ disappears from $\calU(r)$, the three arcs defined by $p_1$, $p_2$, and $p_3$ intersect at a common point $q$, which is equidistant to the three points. Further, since $q$ is currently on $\calU(r)$, there is no red point that is closer to $q$ than $p_i$ for $i=1,2,3$, and the distance from $q$ to each $p_i$, $i=1,2,3$, is equal to the current value of $r$. Hence, $q$ is a vertex of the Voronoi diagram of the red points. This implies that as $r$ changes, the combinatorial structure of $\calU(r)$ does not change until possibly when $r$ is equal to the distance $\lVert q-p \rVert$, where $q$ is a vertex of the Voronoi diagram of all red points and $p$ is a nearest red point of $q$.

Based on the above observation, our algorithm works as follows. We build the Voronoi diagram for all red points, which takes $O(n_r \log n_r)$ time~\cite{ref:FortuneA87,ref:ShamosCl75}. For each vertex $v$ of the diagram, we add $\lVert v-p \rVert$ to the set $\calQ$ (initially $\calQ=\emptyset$), where $p$ is a nearest red point of $v$ ($p$ is available from the diagram).
Note that $|\calQ|=O(n_r)$, and we refer to each value of $\calQ$ as a {\em critical value}. Next, we sort $\calQ$, and then do binary search on $\calQ$ using the decision algorithm to find the smallest value $r_2'$ of $\calQ$ with $r_2'\geq r^*$ as well as the largest value $r_1'$ of $\calQ$ smaller than $r^*$, which can be done in $O(n\log n_r)$ time (note that $n_r\leq n$). By definition, $(r_1',r_2']$ contains $r^*$ and $(r_1',r_2')$ does not contain any value of $\calQ$. According to the above observation, if $r^*\neq r_2'$, then the combinatorial structure of $\calU(r^*)$ is the same as that of $\calU(r)$ for any $r\in (r_1',r_2')$.

We analyze the running time of this subroutine for all instances of $\calI$. Clearly, the total time for all instances is bounded by $O(|\calI|\cdot n\log n)$, which is $O(n^2\log n)$ as $|\calI|=O(n)$. We can reduce the time to $O(n\log n)$ by considering the critical values of all instances of $\calI$ altogether. Specifically, let $\calQ$ now be the set of critical values of all instances of $\calI$. Then, $|\calQ|=O(n)$. We sort $\calQ$ and do binary search on $\calQ$ to find $r_1'$ and $r_2'$ as defined above with respect to the new $\calQ$. Now, for each instance of $\calI$, if $r^*\neq r_2'$, then the combinatorial structure of $\calU(r^*)$ is the same as that of $\calU(r)$ for any $r\in (r_1',r_2')$. The total time for all instances of $\calI$ is now bounded by $O(n\log n)$. Finally, we update $r_1=\max\{r_1,r_1'\}$ and $r_2=\min\{r_2,r_2'\}$. As $r^*\in (r_1',r_2']$, the new interval $(r_1,r_2]$ still contains $r^*$. Further, as $(r_1,r_2)\subseteq (r_1',r_2')$, for each instance of $\calI$, if $r^*\neq r_2$, then the combinatorial structure of $\calU(r^*)$ is the same as that of $\calU(r)$ for any $r\in (r_1,r_2)$.

\subsubsection{Sorting the upper envelope vertices and blue points}
\label{subsubsec:ConnectivitiesBetweenBRPoints}


The goal of the second subroutine is to shrink the interval $(r_1,r_2]$ such that it contains $r^*$ and if $r^*\neq r_2$, then the sorted list of all vertices of $\calU(r^*)$ and all blue points by their $x$-coordinates is the same as the sorted list of all vertices of $\calU(r)$ and all blue points for any $r\in (r_1,r_2)$.

Recall that after the first subroutine, the interval $(r_1,r_2]$ contains $r^*$, and if $r^*\neq r_2$, then the combinatorial structure of $\calU(r^*)$ is the same as that of $\calU(r)$ for any $r\in (r_1,r_2)$.

\begin{figure}[t]
    \centering
    \begin{minipage}{0.48\textwidth}
        \centering
        \includegraphics[height=1.2in]{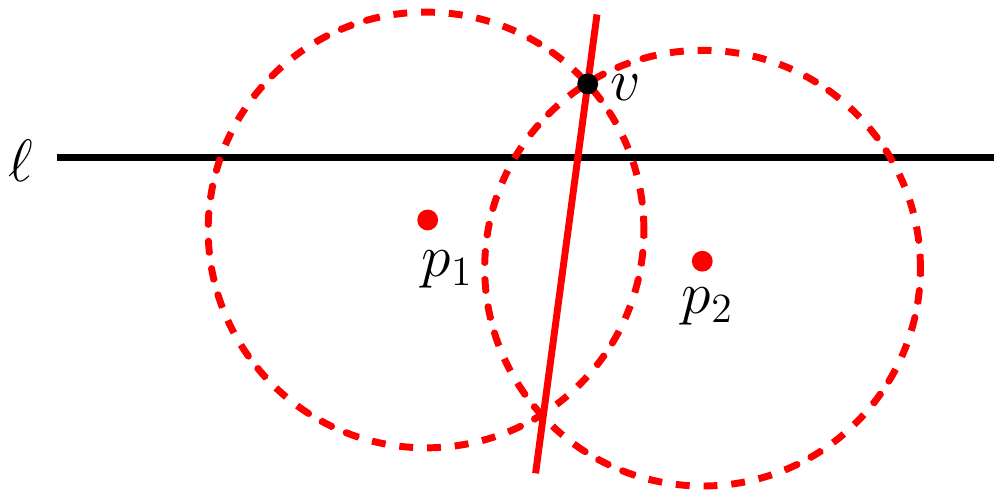}
        \caption{\footnotesize Illustrating a vertex $v$ of the upper envelope, which is defined by two red points $p_1$ and $p_2$. The red solid segment is the bisector of $p_1$ and $p_2$.}
        \label{fig:SortingRedBluePoints}
    \end{minipage}
    \hspace{0.05in}
    \begin{minipage}{0.49\textwidth}
        \centering
        \includegraphics[height=1.3in]{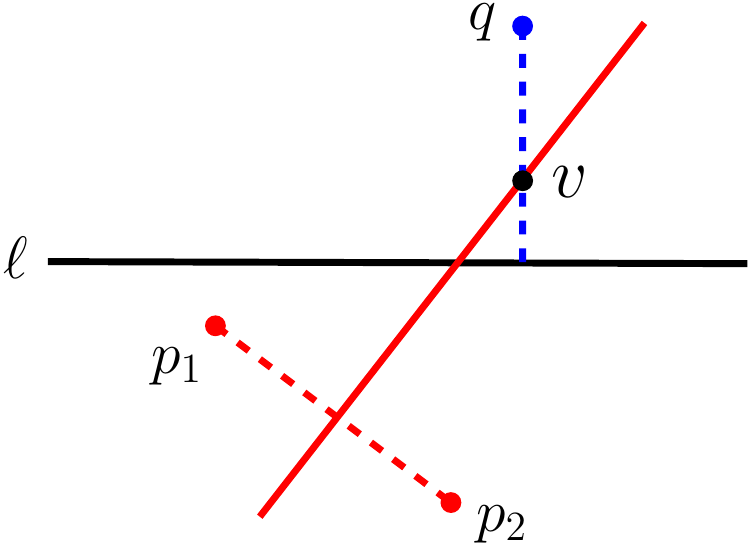}
        \caption{\footnotesize Illustrating the scenario where $x(q)=x(v)$, where $v$ is on the bisector (the red solid segment) of $p_1$ and $p_2$.}
        \label{fig:ResolvingComparisons}
    \end{minipage}
\end{figure}

To sort all vertices of $\calU(r^*)$ and all blue points, we apply Cole's parametric search~\cite{ref:ColeSl87} with AKS sorting network \cite{ref:AjtaiPr83}, using the CS algorithm as the decision algorithm; the running time is bounded by $O(n\log n)$ as the number of vertices of $\calU(r^*)$ is $O(n_r)$ and the number of blue points is $O(n_b)$ (and $n_r+n_b=O(n)$). To see why this works, it suffices to argue that the ``root'' of each comparison involved in the sorting can be obtained in $O(1)$ time (more specifically, the root refers to the value of $r\in (r_1,r_2)$ at which the two operands involved in the comparison are equal). Indeed, the comparisons can be divided into three types based on their operands: (1) a comparison between the $x$-coordinates of two blue points; (2) a comparison between the $x$-coordinates of two vertices of $\calU(r^*)$; (3) a comparison between the $x$-coordinates of a blue point and a vertex of $\calU(r^*)$. For the first type, as blue points are fixed, independent of the parameter $r$, it is trivial to handle. For the second type, as the combinatorial structure of $\calU(r)$ does not change for all $r\in (r_1,r_2)$, each such comparison can be resolved by taking any value of $r\in (r_1,r_2)$ and then comparing the two vertices under $r$. The third type is a little more involved. Consider the comparison of the $x$-coordinates of a blue point $q$ and a vertex $v$ of $\calU(r^*)$. Note that $v$ is the intersection of arcs of two circles of radius $r$ and centered at two red points, say $p_1$ and $p_2$, respectively. Observe that $v$ is on the bisector of $p_1$ and $p_2$ (e.g., see Fig.~\ref{fig:SortingRedBluePoints}). Furthermore, when $r$ changes, $v$ moves on the bisector of $p_1$ and $p_2$, while the position of the blue point $q$ does not change. Hence, the root of the comparison, i.e., the value $r$ (if exists) in $(r_1,r_2)$ such that $x(q)=x(v)$ can be obtained in constant time by elementary geometry (e.g., see Fig.~\ref{fig:ResolvingComparisons}). Note that if such $r$ does not exist in $(r_1,r_2)$, then either $x(q)<x(v)$ holds for all $r\in (r_1,r_2)$ or $x(q)>x(v)$ holds for all $r\in (r_1,r_2)$, which can be easily determined.
As such, with Cole's parametric search~\cite{ref:ColeSl87} and the linear time decision algorithm (i.e., the CS algorithm), we can obtain a sorted list of the upper envelope vertices and the blue points by their $x$-coordinates; the algorithm shrinks the interval $(r_1,r_2]$ so that the new interval $(r_1,r_2]$ contains $r^*$ and if $r^*\neq r_2$, then the above sorted list is fixed for all $r\in (r_1,r_2)$.

Since the running time of the above sorting algorithm is $O(n\log n)$, as before for the first subroutine, the sorting for all problem instances of $\calI$ takes $O(n^2\log n)$ time. To reduce the time, as before, we sort all elements in all instances of $\calI$ altogether, which takes $O(n\log n)$ time in total. Specifically, in each problem instance, we need to sort a set of blue points and vertices of upper envelopes of a set of red points. We put all blue points and the upper envelopes of all red points of all problem instances of $\calI$ in one coordinate system and apply the sorting algorithm as above. One difference is that we now have a new type of comparisons: compare the $x$-coordinate of a vertex $v_1$ of the upper envelope from one problem instance with the $x$-coordinate of a vertex $v_2$ of the upper envelope from another problem instance. In this case, when $r$ changes, both $v_1$ and $v_2$ moves on the bisectors of their defining red points. But we can still find in constant time a root $r$ (if exists) in $(r_1,r_2)$ for the comparison by elementary geometry. As such, we can complete the sorting for all problem instances of $\calI$ in $O(n\log n)$ time in total, for the total number of all blue points and red points in all problem instances of $\calI$ is $O(n)$. Again, the interval $(r_1,r_2]$ will be shrunk. This finishes the second subroutine.

\subsubsection{Deciding whether each blue point is below the upper envelope}
\label{subsubsec:DecidingWhetherEachBluePoint}

We now have an interval $(r_1,r_2]$ containing $r^*$ such that if $r^*\neq r_2$, then each blue point $q$ is spanned by an arc $\alpha_q(r)$ of $\calU(r)$ defined by the same red point for all $r\in (r_1,r_2)$ (note that the arc $\alpha_q(r)$ moves as $r$ changes, for $r$ is the radius of the arc). Each blue point $q$ is below the upper envelope $\calU(r)$ if and only if $q$ is below the arc $\alpha_q(r)$. The goal of the third subroutine is to shrink the interval $(r_1,r_2]$ so that the new interval $(r_1,r_2]$ still contains $r^*$ and if $r^*\neq r_2$, then for each blue point $q$, the relative position of $q$ with respect to $\alpha_q(r)$ (i.e., whether $q$ is above or below $\alpha_q(r)$) is fixed for all $r\in (r_1,r_2)$. To this end, we proceed as follows.

As $r$ changes in $(r_1,r_2)$, $\alpha_q(r)$ changes while $q$ does not. For each blue point $q$, we compute in constant time a critical value $r$ (if exists) in $(r_1,r_2)$ such that $q$ is on $\alpha_q$, and we add $r$ to the set $\calQ$ ($\calQ=\emptyset$ initially). Note that if such value $r$ does not exist in $(r_1,r_2)$, then either $q$ is above $\alpha_q(r)$ for all $r\in (r_1,r_2)$ or $q$ is below $\alpha_q(r)$ for all $r\in (r_1,r_2)$, which can be easily determined. The size of $\calQ$ is at most $n_b$. Then, we sort $\calQ$, and do binary search on $\calQ$ with our decision algorithm to find the smallest value $r_2'$ of $\calQ$ with $r_2'\geq r^*$ and the largest value $r_1'$ of $\calQ$ with $r_1'<r^*$. We then update $r_1=\max\{r_1,r_1'\}$ and $r_2=\min\{r_2,r_2'\}$. The new interval $(r_1,r_2]$ still contains $r^*$ and $(r_1,r_2)$ does not contain any value of $\calQ$. Hence, if $r^*\neq r_2$, then for each blue point $q$, the relative position of $q$ with respect to $\alpha_q(r)$ is fixed for all $r\in (r_1,r_2)$. As such, the new interval $(r_1,r_2]$ satisfies the goal of the third subroutine as mentioned above.

Finally, we pick an arbitrary $r\in (r_1,r_2)$, and for each blue point $q$, if $q$ is below the arc $\alpha_q(r)$, then we add $q$ to the set $S_{i}$.

The running time of the above algorithm is $O(n\log n_b)$. Thus the total time of the third subroutine is $O(n^2\log n)$ for all problem instances of $\calI$. To reduce the time, we again consider the subroutine of all instances of $\calI$ altogether. More specifically, we put all critical values $r$ in all problem instances of $\calI$ in $\calQ$. Thus, the size of $\calQ$ is $O(n)$. We then run the same algorithm as above using the new set $\calQ$. The total time is bounded by $O(n\log n)$.

\subsubsection{Terminating the algorithm}

This finishes the $i$-th step of the BFS, which computes a set $S_i$ along with an interval $(r_1,r_2]$. According to the above discussion, $(r_1,r_2]$ contains $r^*$ and if $r^*\neq r_2$ (and thus $r^*\in (r_1,r_2)$), then $S_i=S_i(r^*)=S_i(r)$ for all $r\in (r_1,r_2)$.

If the point $t$ is in $S_i$ and $i\leq \lambda$, then we stop the algorithm. In this case, we have the following lemma.

\begin{lemma}\label{lem:30}
If $t\in S_i$ and $i\leq \lambda$, then $r^*=r_2$.
\end{lemma}
\begin{proof}
Assume to the contrary that $r^*\neq r_2$. Then, since $r^*\in (r_1,r_2]$, we have $r^*\in (r_1,r_2)$. Let $r'=(r_1+r^*)/2$. Clearly, $r'\in (r_1,r_2)$ and $r'<r^*$. As $r'\in (r_1,r_2)$, $S_i=S_i(r')$ by our algorithm invariant. Since $t\in S_i(r')$, we obtain that $d_{r'}(s,t)=i \leq \lambda$. This incurs contradiction as $r'<r^*$ and  $r^*$ is the minimum value $r$ with $d_{r}(s,t)\leq \lambda$.
\end{proof}

If $t\not\in S_i$ and $i=\lambda$, then we also stop the algorithm.   In this case, we have the following lemma.

\begin{lemma}\label{lem:35}
If $t\not\in S_i$ and $i=\lambda$, then $r^*=r_2$.
\end{lemma}
\begin{proof}
Assume to the contrary that $r^*\neq r_2$. Then, $r^*\in (r_1,r_2)$, for $r^*\in (r_1,r_2]$. By our algorithm invariant, $S_j=S_j(r)$ for all $r\in (r_1,r_2)$ and for all $j\leq i$. Hence, $S_j=S_j(r^*)$ for all $j\leq i$. As $t\not\in S_i$, according to our algorithm, $t\not\in \bigcup_{j=0}^i S_j$. Therefore, $t\not\in \bigcup_{j=0}^i S_j(r^*)$, implying that $d_{r^*}(s,t)>i=\lambda$. However, by the definition of $r^*$, $d_{r^*}(s,t)\leq \lambda$ holds. We thus obtain contradiction.
\end{proof}

Since initially $i=0$ and $S_0=\{s\}$, the above implies that the BFS algorithm will stop in at most $\lambda$ steps.
As each step takes $O(n\log n)$ time, the value $r^*$ can be computed in $O(\lambda\cdot n\log n)$ time.

\begin{theorem}\label{theorem:DsppAlgorithm}
    The reverse shortest path problem for $L_2$ unweighted unit-disk graphs can be solved in $O(\lfloor \lambda \rfloor \cdot n \log n)$ time.
\end{theorem}

\section{The $L_2$ unweighted case -- the second algorithm}
\label{sec:second}

In this section, we present our second algorithm for the $L_2$ unweighted RSP problem. As discussed in Section~\ref{subsec:OurApproach}, the main idea is to somehow combine the strategies of the first unweighted RSP algorithm in Section~\ref{sec:first} and the naive binary search algorithm using the $L_2$ distance selection algorithm~\cite{ref:KatzAn97}.

First of all, we still build in $O(n\log n)$ time the grid $\Psi(P)$ as in Section~\ref{sec:grid}, and thus the information of Lemma~\ref{lem:20} is available for the grid. More specifically, we obtain an interval $(r_1,r_2]$ such that if $r^*\neq r_2$, then the combinatorial data structure of $\Psi_r(P)$ is fixed for all $r\in (r_1,r_2)$, implying that $\calC$, $P'$, $N(C)$ and $P(C)$ for each $C\in \calC$ are fixed for all $r\in (r_1,r_2)$. Next, we will run the BFS algorithm, but in a different way than before.

We partition the cells of $\calC$ into {\em large cells} and {\em small cells}: a cell $C$ is a large cell if $|P(C)| \geq (n / \log n)^{3/4}$ and is a small cell otherwise. Thus the number of large cells is at most $n^{1/4} \log^{3/4} n$. For all pairs of cells $(C,C')$ with $C\in \calC$ and $C'\in N(C)$, we call $(C,C')$ a {\em small-cell pair} if both $C$ and $C'$ are small cells and a {\em large-cell pair} otherwise (i.e., at least one cell is a large cell). As $|N(C)|=O(1)$ for each cell $C$ and the number of large cells is at most $n^{1/4} \log^{3/4} n$, the total number of large-cell pairs is $O(n^{1/4} \log^{3/4} n)$.

Recall that each step of the BFS algorithm of our first algorithm in Section~\ref{subsec:bfs} boils down to solving instances of Subproblem~\ref{subpro:10}, and each such instance involves a cell pair $(C,C')$ with $C\in \calC$ and $C'\in N(C)$. If $(C,C')$ is a large-cell pair, we will run the same algorithm as in Section~\ref{subsec:bfs}. Otherwise, we will use the original CS algorithm to solve it, which takes only linear time.
For this, with the help of the $L_2$ distance selection algorithm~\cite{ref:KatzAn97}, we preprocess all these small-cell pairs before starting the BFS algorithm by the following lemma.

\begin{lemma}\label{lem:40}
An interval $(r_1',r_2']$ containing $r^*$ can be computed in $O(n^{5/4}\log^{7/4} n)$ time with the following property: if $r^* \neq r_2'$, then for any $r \in (r_1', r_2')$, for any small-cell pair $(C,C')$ with $C\in \calC$ and $C'\in N(C)$, an edge connects a point $p \in P(C)$ and a point $p' \in P(C')$ in $G_r(P)$ if and only if an edge connects $p$ and $p'$ in $G_{r^*}(P)$.
\end{lemma}
\begin{proof}
Let $\Pi$ denote the set of all small-cell pairs $(C,C')$ with $C\in \calC$ and $C'\in N(C)$. We use $(C_i,C_i')$ to denote the $i$-th pair of $\Pi$; let $P_i$ denote the set of points of $P$ in the two cells $C_i$ and $C_i'$, and let $n_i=|P_i|$. Let $m=|\Pi|$. Note that $m=O(n)$. By the definition of small cells, we have $n_i\leq 2\cdot (n / \log n)^{3/4}$. Since $|N(C)|=O(1)$ for each cell $C$, it holds that $\sum_{i=1}^m n_i=O(n)$.
For each $P_i$, let $D_i$ denote the set of distances of all pairs of points of $P_i$. Hence, $|D_i|=n_i(n_i-1)/2$. Define $\calD=\bigcup_{i=1}^m D_i$.

Let $r_2'$ be the smallest value of $\calD$ with $r_2'\geq r^*$ and let $r_1'$ be the largest value of $\calD$ smaller than $r^*$. By definition, $(r_1',r_2']$ contains $r^*$ and the open interval $(r_1',r_2')$ does not contain any value of $\calD$ and thus any value of $D_i$ for each $i$. Therefore, for any two points $p$ and $p'$ of $P_i$, either $\lVert p-p' \rVert <r$ holds for all $r\in (r_1',r_2')$ or $\lVert p-p' \rVert >r$ holds for all $r\in (r_1',r_2')$.
Thus, $(r_1',r_2']$ satisfies the lemma statement.
In the following, we only describe the algorithm for finding $r_2'$ since the algorithm for finding $r_1'$ is similar.

For convenience, for any $r$, we say that $r$ is {\em feasible} if $r\geq r^*$ and {\em infeasible} otherwise. Note that
if $r$ is a feasible value, then $r'$ is also feasible for any $r'>r$; symmetrically, if $r$ is infeasible, then $r'$ is also infeasible for any $r'<r$.
Recall that given any $r$, we can decide whether $r\geq r^*$ in linear time using the decision algorithm (i.e., the CS algorithm).

For each $P_i$, we use the $L_2$ distance selection algorithm~\cite{ref:KatzAn97} to compute the median distance of $D_i$, denoted by $d_i$, which takes $O(n_i^{4/3}\log^2 n_i)$ time. Then, we sort all these medians $d_i$, $1\leq i\leq m$, and do binary search on the sorted list using the decision algorithm. In $O(n\log n)$ time, we can determine whether each $d_i$ is feasible. Among all these medians, we keep the smallest feasible value, denoted by $d^1$. This finishes the first round of the algorithm.

In the second round, for each $d_i$, if it is feasible, then any value of $D_i$ larger than $d_i$ is also feasible; in this case, we compute the $(|D_i|/4)$-th smallest value of $D_i$, denoted by $d_i'$. If $d_i$ is infeasible, then any value of $D_i$ smaller than $d_i$ is also infeasible; in this case, we compute the $(3|D_i|/4)$-th smallest value of $D_i$, denoted by $d_i'$. Next, we determine whether the values $d_i'$ are feasible for all $1\leq i\leq m$ in the same way as above (i.e., doing binary search using the decision algorithm); we keep the smallest feasible value, denoted by $d^2$.

We then continue the next round in a similar way as above. After $O(\log n)$ rounds, the values of all sets $D_i$ are processed and we obtain a set of $O(\log n)$ feasible values $d^1$, $d^2$, \ldots; among all these values, the smallest one is $r_2'$.


For the time analysis, the algorithm has $O(\log n)$ rounds and each round takes $O(n\log n + \sum_{i=1}^mn_i^{4/3}\log^2 n_i)$ time. Since $n_i\leq 2\cdot (n / \log n)^{3/4}$ for each $1\leq i\leq m$, and $\sum_{i=1}^{m}n_i=O(n)$, the sum $\sum_{i=1}^mn_i^{4/3}$ achieves maximum when each $n_i$ is equal to $2\cdot (n / \log n)^{3/4}$ (and thus $m=O(n^{1/4} \log^{3/4} n)$). Hence, $\sum_{i=1}^mn_i^{4/3}=O(n^{5/4} / \log^{1/4} n)$. Therefore, each round of the algorithm takes $O(n^{5/4}\log^{7/4} n)$ time, which is dominated by the $L_2$ distance selection algorithm~\cite{ref:KatzAn97}.
The total time of the algorithm is thus $O(n^{5/4}\log^{11/4} n)$.

In what follows, we reduce the runtime of the algorithm by a logarithmic factor. The new algorithm still has $O(\log n)$ rounds. The difference is that instead of applying the $L_2$ distance selection algorithm~\cite{ref:KatzAn97} directly, we only use a subroutine of that algorithm. This also simplifies the overall algorithm. To avoid the lengthy background discussion, we use concepts from~\cite{ref:KatzAn97} without further explanation (refer to the initial version of the algorithm in Section~4~\cite{ref:KatzAn97} for the details).

Initially, we set $I_0=(0,\infty]$; we also add $\infty$ to $\calD$. Given an interval $I_{j-1}=(a_{j-1},b_{j-1}]$ that contains $r^*$ with $b_{j-1}\in \calD$, the $j$-th round of the algorithm produces an interval $I_j=(a_{j},b_{j}]$ that also contains $r^*$ with $b_j\in \calD$ such that $I_j\subseteq I_{j-1}$ and the number of values of $\calD$ contained in $I_j$ is only a constant fraction of the number of values of $\calD$ contained in $I_{j-1}$. Thus, after $O(\log n)$ rounds, we are left with a sufficiently small number of distances of $\calD$, from which it is trivial to find $r_2'$.

The $j$-th round of the algorithm works as follows. For each set $P_i$, we compute a compact representation of all pairs of points of $P_i$ whose distances lie in $I_{i-1}$, which can be done in $O(n_i^{4/3}\log n_i)$ time~\cite{ref:KatzAn97}. Such a compact representation is a collection of $O(n_i^{4/3})$ complete bipartite graphs $\{Q_k \times W_k\}_k$, where both $\sum_{k} |Q_k|$ and  $\sum_k |W_k|$ are bounded by $O(n_i^{4/3} \log n_i)$.
For each $k$, the distance between any point in $Q_k$ and any point of $W_k$ is in $I_{i-1}$. Next, we replace each complete bipartite graph $Q_k \times W_k$ by a set $E_k$ of expander graphs whose total number of edges is $O(|Q_k| + |W_k|)$. Then the total number of edges of all sets of expander graphs $\{E_k\}_k$ is $\sum_k O(|Q_k| + |W_k|) = O(n_i^{4/3} \log n_i)$. Each edge of an expander graph is associated with a distance of two points corresponding to the two nodes of the graph it connects. Let $L_{i}$ denote the set of distances of all edges in all expander graphs of $\{E_k\}_k$; the size of $L_{i}$ is $O(n_i^{4/3} \log n_i)$. Let $\calL$ denote the union of all such $L_i$'s. Then, $|\calL|=\sum_{i=1}^{m}n_i^{4/3}\log n_i$, which is bounded by $O(n^{5/4}\log^{3/4} n)$ as discussed above.
By doing binary search with the decision algorithm on $\calL$, we can compute the smallest feasible value $b_j$ and the largest infeasible value $a_j$ of $\calL$. Hence, $(a_j,b_j]$ contains $r^*$ and $(a_j,b_j)$ does not contain any value of $\calL$. Note that when doing binary search on $\calL$, we do not need to sort it first; instead we use the linear time selection algorithm~\cite{ref:BlumTi73}. As such, finding $a_j$ and $b_j$ can be done in $O(n^{5/4}\log^{3/4} n)$ time, which is also the total time of this round. Let $I_j=(a_j,b_j]$.
The analysis of~\cite{ref:KatzAn97} shows that the total number of values of $\calD$ in $I_j$ is a constant fraction of the total number of values of $\calD$ in $I_{j-1}$.

As the algorithm has $O(\log n)$ rounds and each round runs in $O(n^{5/4}\log^{3/4} n)$ time, the overall time of the algorithm is $O(n^{5/4}\log^{7/4} n)$.
\end{proof}

With the interval $(r_1',r_2']$ computed by the above lemma, we update $r_1=\max\{r_1,r_1'\}$ and $r_2=\min\{r_2,r_2'\}$. By definition, $r^*\in (r_1,r_2]\subseteq (r'_1,r_2']$. Hence, the interval $(r_1,r_2]$ also has the same property as $(r_1',r_2']$ in Lemma~\ref{lem:40}.

Next, we run the BFS algorithm as in Section~\ref{subsec:bfs}. To solve each instance of Subproblem~\ref{subpro:10}, if one of the two involved cells is a large cell (we refer to this case as the {\em large-cell instance}), then we use the same algorithm as before, i.e.,  parametric search; otherwise (i.e., both involved cells are small cells; we refer to this case as {\em small-cell instance}), due to the preprocessing of Lemma~\ref{lem:40}, we can solve the subproblem directly using the original CS algorithm by picking an arbitrary value $r\in (r_1,r_2)$. In this way, the time for solving all small-cell instances in the entire BFS algorithm is $O(n)$. For each large-cell instance, it can be solved in $O(n\log n)$ time as discussed in Section~\ref{subsec:bfs}. As the number of large cells of $\calC$ is at most $n^{1/4} \log^{3/4} n$ and $|N(C)|=O(1)$ for each cell $C\in \calC$, the total number of large-cell instances of Subproblem~\ref{subpro:10} is at most $O(n^{1/4} \log^{3/4} n)$. Hence, the total time for solving the large-cell instances in the entire BFS algorithm is $O(n^{5/4}\log^{7/4} n)$.
The proof of the following lemma presents the details of the new BFS algorithm sketched above.

\begin{lemma}
The BFS algorithm, which computes $r^*$, can be implemented in $O(n^{5/4}\log^{7/4} n)$ time.
\end{lemma}
\begin{proof}
We define $S_i$ and $S_i(r)$ in the same way as in Section~\ref{subsec:bfs}. Initially, we set $S_0=\{s\}$. Before the $i$-step starts, we have an interval $(r_1,r_2]$.
Again, the algorithm maintains an invariant that the $i$-th step shrinks $(r_1,r_2]$ so that it contains $r^*$ and if $r^*\neq r_2$, then $S_i=S_i(r^*)=S_i(r)$ for any $r\in (r_1,r_2)$. Initially, the invariant trivially holds for $S_0$.

Consider the $i$-th step. Assume that the invariant holds for $S_{i-1}$, i.e., we have an interval $(r_1,r_2]$ containing $r^*$ such that if $r^*\neq r_2$, then $S_{i-1}=S_{i-1}(r)=S_{i-1}(r^*)$ for any $r\in (r_1,r_2)$, and $S_{i-1}$ is available to us. Using the grid information of $\Psi(P)$, we obtain the grid cells containing the points of $S_{i-1}$. For each such cell $C$, as before in Section~\ref{subsec:bfs}, we add  to $S_i$ the points of $P\cap C$ that have not been discovered yet. Then, for each neighbor $C'$ of $C$, we need to solve Subproblem~\ref{subpro:10}; we use $\calI$ to denote the set of instances of this subproblem in this step.

Consider two cells $C$ and $C'$ involved in an instance of $\calI$. If one of them is a large cell, then we run the same parametric search algorithm as in Section~\ref{subsec:bfs}, i.e., the three subroutines. As before, the time of the algorithm is bounded by $O(n\log n)$ and the algorithm shrinks the interval $(r_1,r_2]$ so that the algorithm invariant is maintained. Recall that in Section~\ref{subsec:bfs} we solve all problem instances in each step of the BFS algorithm altogether. Here instead it suffices to solve each problem instance individually. As the number of large cells is at most $O(n^{1/4} \log^{3/4} n)$, the total number of large-cell instances in the entire BFS algorithm is $O(n^{1/4} \log^{3/4} n)$. Hence, the total time for solving the large-cell instances of Subproblem~\ref{subpro:10} in the entire BFS is $O(n^{5/4}\log^{7/4} n)$.

We now consider the small-cell instance where both $C$ and $C'$ are small cells.
Note that in each instance of Subproblem~\ref{subpro:10}, all red points are in one cell, say, $C$, and all blue points are in the other cell $C'$.
Let $P_R$ be the set of red points in $C$ and $P_B$ be the set of blue points in $C'$.
According to Lemma~\ref{lem:40}, if $r^*\neq r_2$ (and thus $r^*\in (r_1,r_2)$), then for any point $p\in P_R$ and any point $p'\in P_B$, either $\lVert p-p' \rVert < r$ holds for all $r\in (r_1,r_2)$ or $\lVert p-p' \rVert > r$ holds for all $r\in (r_1,r_2)$, implying that $\lVert p-p' \rVert > r^*$ if and only if $\lVert p-p' \rVert > r$ for any $r\in (r_1,r_2)$.
Therefore, we can solve the subproblem in the following way. We first take any $r\in (r_1,r_2)$. Then we run the CS algorithm to solve the subproblem with $r$ as the radius, which takes $O(n_r+n_b)$ time. Note that the interval $(r_1,r_2]$ will not be changed in this case. Due to the preprocessing in  Lemma~\ref{lem:40}, the algorithm invariant still holds (i.e., $(r_1,r_2]$ contains $r^*$ and if $r^*\neq r_2$, then $S_i=S_i(r^*)=S_i(r)$ for any $r\in (r_1,r_2)$). The total time for solving the small-cell instances in the entire BFS is $O(n)$ because as in the CS algorithm each cell will be involved in at most $O(1)$ instances of the subproblem in the entire BFS algorithm.

After the $i$-th step, as before, we obtain the set $S_i$ and an interval $(r_1,r_2]$ containing $r^*$ such that if $r^*\neq r_2$, then $S_i=S_i(r^*)=S_i(r)$ for any $r\in (r_1,r_2)$. If $t\in S_i$ and $i\leq \lambda$, then we can stop the algorithm; by Lemma~\ref{lem:30}, we have $r^*=r_2$. If $t\not\in S_i$ and $i=\lambda$, we also stop the algorithm; by Lemma~\ref{lem:35}, we have $r^*=r_2$.

In summary, the overall time of the BFS algorithm is $O(n^{5/4}\log^{7/4} n)$.
\end{proof}

Combining with the algorithm of Lemma~\ref{lem:40}, the overall time of the algorithm for computing $r^*$ is $O(n^{5/4}\log^{7/4} n)$. We thus obtain the following theorem.

\begin{theorem}
The reverse shortest path problem for $L_2$ unweighted unit-disk graphs can be solved in $O(n^{5/4} \log^{7/4} n)$ time.
\end{theorem}

\section{The $L_2$ weighted case}
\label{sec:WeightedRSPalgorithm}

We follow the notation introduced in Section~\ref{sec:introduction} and Section~\ref{sec:pre}, e.g., $P$, $G_r(P)$, $d_r(s,t)$, and $r^*$, but now defined for weighted unit-disk graphs. Our goal is to compute $r^*$. As discussed in Section~\ref{subsec:OurApproach}, our algorithm utilizes parametric search by parameterizing the WX algorithm~\cite{ref:WangNe20}. We begin with a review of the WX algorithm.



\subsection{A review of the WX algorithm}
\label{subsec:WX}

Given $P$, $r$, and a source point $s\in P$, the WX algorithm can compute shortest paths from $s$ to all points of $P$ in the weighted unit-disk graph $G_r(P)$, and the algorithm runs in $O(n\log^2 n)$ time.

For any point $p$ in the plane, let $\bigodot_p$ denote the disk centered at $p$ with radius $r$.

The first step is to implicitly build a grid $\Psi_r(P)$ of square cells whose side lengths are $r/\sqrt{2}$. For simplicity of discussion, we assume that every point of $P$ lies in the interior of a cell of $\Psi_r(P)$. A \emph{patch} of $\Psi_r(P)$ refers to a square area consisting of $5 \times 5$ cells. For a point $p \in P$, we use $\square_p$ to denote the cell of $\Psi_r(P)$ containing $p$ and use $\boxplus_p$ to denote the patch whose central cell is $\square_p$ (e.g., see Fig.~\ref{fig:patch}).
We refer to cells of $\boxplus_p\setminus \square_p$ as the {\em neighboring cells} of $\square_p$.
As the side length of each cell of $\Psi_r(P)$ is $r/\sqrt{2}$, any two points of $P$ in a single cell of $\Psi_r(P)$ must be connected by an edge in $G_r(P)$. Moreover, if an edge connects two points $p$ and $q$ in $G_r(P)$, then $q$ must lie in $\boxplus_p$ and vice versa.
For any subset $Q \subseteq P$ and a cell $\square$ (resp.,a patch $\boxplus$) of $\Psi_r(P)$, define $Q_{\square} = Q \cap \square$ (resp., $Q_{\boxplus} = Q \cap \boxplus$). The step of implicitly building the grid actually computes the subset $P_\square$ for each cell $\square$ of $\Psi_r(P)$ that contains at least one point of $P$ as well as associate pointers to each point $p \in P$ so that given any $p\in P$, the list of points of $P_{\square_p}$ (resp., $P_{\boxplus_p}$) can be accessed immediately. Building $\Psi_r(P)$ implicitly as above can be done in $O(n \log n)$ time, e.g., by the algorithm of Lemma~\ref{lem:grid}. 

\begin{figure}[t]
    \centering
    \includegraphics[width=2.0in]{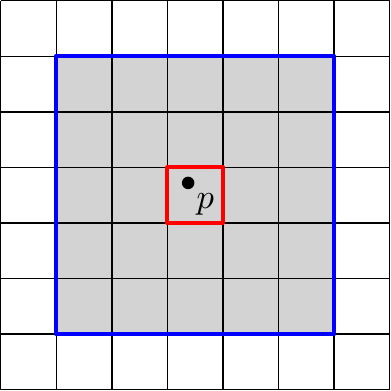}
    \caption{The red cell that contains the point $p$ is $\square_p$ and the square area bounded by blue segments is the patch $\boxplus_p$. All adjacent vertices of $p$ in $G_r(P)$ must lie in the grey region.}
    \label{fig:patch}
\end{figure}

The WX algorithm follows the basic idea of Dijkstra's algorithm and computes an array $dist[\cdot]$ for each point $p \in P$, where $dist[p]$ will be equal to $d_r(s,p)$ when the algorithm terminates.
Different from Dijkstra's shortest path algorithm, which picks a single vertex in each iteration to update the shortest path information of other adjacent vertices, the WX algorithm aims to update in each iteration the shortest path information for all points within one single cell of $\Psi_r(P)$ and pass on the shortest path information to vertices lying in the neighboring cells.

A key subroutine used in the WX algorithm is \textsc{Update}($U$, $V$), which updates the shortest path information for a subset $V \subseteq P$ of points by using the shortest path information of another subset $U \subseteq P$ of points. Specifically, the subroutine finds, for each $v \in V$, $q_v = \arg \min_{u \in U \cap \bigodot_v} \{dist[u] + \lVert u - v \rVert\}$ and update $dist[v] = \min \{dist[v], dist[q_v] + \lVert q_v - v \rVert\}$.

With the subroutine \textsc{Update}($U$, $V$) in hand, the WX algorithm works as follows (refer to Algorithm~\ref{algorithm:WX} for the pseudocode).



\begin{algorithm}[htbp]
    \DontPrintSemicolon
    \caption{The WX Algorithm~\cite{ref:WangNe20}}
    \label{algorithm:WX}

    \SetKwFunction{FSSSP}{WX}

    \SetKwProg{Fn}{Function}{:}{end}

    \Fn{\FSSSP{$P$, $s$}}
    {
        \For{each $p \in P$}
        {
            $dist[p] = \infty$ \;
        }
        $dist[s] = 0$ \;
        $Q = P$ \;
        \While{$Q \neq \emptyset$}
        {
            $z = \arg \min_{p \in Q} \{dist[p]\}$ \; 
            \label{line:findmin}
            \textsc{Update}$(Q_{\boxplus_{z}}, Q_{\square_{z}})$ \tcp{first update}
            \textsc{Update}$(Q_{\square_{z}}, Q_{\boxplus_{z}})$ \tcp{second update}

            $Q = Q \setminus Q_{\square_{z}}$ \;
            \label{line:remove}
        }

        \Return $dist[\cdot]$ \;
    }

\end{algorithm}



Initially, we set $dist[s] = 0$, $dist[p] = \infty$ for all other points $p \in P\setminus\{s\}$, and $Q=P$. Then we enter the main (while) loop. In each iteration, we find a point $z$ with minimum $dist$-value from $Q$, and then execute two update subroutines \textsc{Update}($Q_{\boxplus_z}$, $Q_{\square_z}$) and \textsc{Update}($Q_{\square_z}$, $Q_{\boxplus_z}$). Next, points of $Q_{\square_z}$ are removed from $Q$, because it can be shown that $dist[p]$ for all points $p\in Q_{\square_z}$ have been correctly computed~\cite{ref:WangNe20}. The algorithm stops once $Q$ becomes $\emptyset$.

The efficiency of the algorithm hinges on the implementation of the two update subroutines. We give some details below, which are needed in our RSP algorithm as well.

\subsubsection{The first update}
\label{subsubsec:TheFirstUpdate}

For the first update \textsc{Update}($Q_{\boxplus_z}$, $Q_{\square_z}$), the crucial step is finding a point $q_v \in Q_{\boxplus_z} \cap \bigodot_v$ for each point $v \in Q_{\square_z}$ such that $dist[q_v] + \lVert q_v - v \rVert$ is minimized. If we assign $dist[q]$ as a weight to each point $q \in Q_{\boxplus_z}$, then the problem is equivalent to finding the additively-weighted nearest neighbor $q_v$ from $Q_{\boxplus_z} \cap \bigodot_v$ for each $v \in Q_{\square_z}$. To this end, Wang and Xue~\cite{ref:WangNe20} proved a {\em key observation} that any point $q \in Q_{\boxplus_z}$ that minimizes $dist[q] + \lVert q - v \rVert$ must lie in $\bigodot_v$. This implies that for each point $v \in Q_{\square_z}$, its additively-weighted nearest neighbor in $Q_{\boxplus_z}$ is also its additively-weighted nearest neighbor in $Q_{\boxplus_z} \cap \bigodot_v$.
As such, $q_v$ for all $v \in Q_{\square_z}$ can be found by first building an additively-weighted Voronoi Diagram on points of $Q_{\boxplus_z}$~\cite{ref:FortuneA87} and then performing point locations for all $v \in Q_{\square_z}$~\cite{ref:EdelsbrunnerOp86,ref:KirkpatrickOp83,ref:SarnakPl86}. In this way, since $\sum_{z_i}|P_{\boxplus_{z_i}}|=O(n)$, where $z_i$ refers to the point $z$ in the $i$-th iteration of the main loop, the first updates for all iterations of the main loop can be done in $O(n \log n)$ time in total~\cite{ref:WangNe20}.

\subsubsection{The second update}
\label{subsubsec:TheSecondUpdate}

The second update \textsc{Update}($Q_{\square_z}$, $Q_{\boxplus_z}$) is more challenging because the above key observation no longer holds. Since $Q_{\boxplus_z}$ has $O(1)$ cells of $\Psi_r(P)$, it suffices to perform \textsc{Update}$(Q_{\square_{z}}, Q_{\square})$ for all cells $\square\in \boxplus_z$.

If $\square$ is $\square_z$, then $Q_{\square_z} = Q_{\square}$. Since the distance between any two points in $\square_z$ is at most $r$, we can easily implement \textsc{Update}$(Q_{\square_{z}}, Q_{\square})$ in $O(|Q_{\square_z}| \log |Q_{\square_z}|)$ time, by first building a additively-weighted Voronoi diagram on points of $Q_{\square_z}$ (each point $q\in Q_{\square_z}$ is assigned a weight equal to $dist[q]$), and then using it to find the additively-weighted nearest neighbor $q_v$ for each point $v\in Q_{\square_z}$.

If $\square$ is not $\square_z$, a useful property is that $\square$ and $\square_z$ are separated by an axis-parallel line. The WX algorithm implements \textsc{Update}$(Q_{\square_{z}}, Q_{\square})$ with the following three steps. Let $U=Q_{\square_{z}}$ and $V=Q_{\square}$.

\begin{enumerate}
    \item Sort points of $U$ as $\{u_1, u_2, ..., u_{|U|}\}$ such that $dist[u_1] \leq dist[u_2] \leq ... \leq dist[u_{|U|}]$.
    \item Compute $|U|$ disjoint subsets $\{V_1, V_2, ..., V_{|U|}\}$ with $V_i = \{v \in V\ |\ v \in \bigodot_{u_i} \; \text{and} \; v \notin \bigodot_{u_j} \; \text{for all} \; 1 \leq j < i \}$. Equivalently, for each point $v \in V$, $v$ is in $V_{i_v}$, where $i_v$ is the smallest index $i$ (if exists) such that $\bigodot_{u_i}$ contains $v$.

    \item Initialize $U' = \emptyset$. Proceed with $|U|$ iterations for $i = |U|, |U| - 1, ..., 1$ sequentially and do the following in each iteration for $i$: (1) Add $u_i$ to $U'$;
        (2) for each point $v \in V_i$,
        compute $q_v = \arg \min_{u \in U'} \{dist[u] + \lVert u - v \rVert\}$;
         (3) update $dist[v] = \min\{dist[v], dist[q_v] + \lVert q_v - v \rVert\}$.
\end{enumerate}

By the definition of $V_i$, $U \cap \bigodot_v \subseteq U' = \{u_{|U|}, u_{|U| - 1}, ..., u_i\}$ for each $v \in V_i$ in the iteration for $i$ of Step 3. Wang and Xue~\cite{ref:WangNe20} proved that $q_v$ found for each $v \in V_i$ in Step 3 must lie in $\bigodot_{v}$. They gave a method to implement Step 2 in $O(k \log k)$ time by making use of the property that $U$ and $V$ are separated by an axis-parallel line,
where $k = |U| + |V|$.
Step 3 can be considered as an offline insertion-only additively-weighted nearest neighbor searching problem and the WX algorithm solves the problem in $O(k\log^2 k)$ time using the standard logarithmic method~\cite{ref:BentleyDe79}, with $k = |U| + |V|$.

As such, the second updates for all iterations in the WX algorithm takes $O(n \log^2 n)$ time in total~\cite{ref:WangNe20}, which dominates the entire algorithm (other parts of the algorithm together takes $O(n\log n)$ time).


\subsection{The RSP algorithm}
\label{subsec:WeightedRSPalgorithm}

We now tackle the RSP problem, i.e., given $\lambda$ and $s,t\in P$, compute $r^*$. We will ``parameterize'' the WX algorithm reviewed above.

Recall that the decision problem is to decide whether $r^*\leq r$ for a given $r$. Notice that $r^*\leq r$ holds if and only if $d_r(s,t)\leq \lambda$. The decision problem can be solved in $O(n\log^2 n)$ time by running the WX algorithm on $r$. In the following, we refer to the WX algorithm as the {\em decision algorithm}. We say that $r$ is a {\em feasible value} if $r^*\leq r$ and an {\em infeasible value} otherwise.

As discussed in Section~\ref{subsec:OurApproach}, to find $r^*$, we run the decision algorithm with a parameter $r$ in an interval $(r_1, r_2]$ by simulating the algorithm on the unknown $r^*$.
The interval always contains $r^*$ but will be shrunk during course of the algorithm (for simplicity, when we say $(r_1, r_2]$ is shrunk, this also include the case that $(r_1, r_2]$ does not change).
Initially, we set $r_1=0$ and $r_2=\infty$.



The first step is to build a grid for $P$. The goal is to shrink $(r_1, r_2]$ so that it contains $r^*$ and if $r^*\neq r_2$ (and thus $r^*\in (r_1,r_2)$), for any $r\in (r_1,r_2)$, the grid $\Psi_r(P)$ has the same combinatorial structure as $\Psi_{r^*}(P)$ in the following sense: (1) Both grids have the same number of rows and columns; (2) for any point $p \in P$, $p$ lies in the $i$-th row and $j$-th column of $\Psi_{r}(P)$ if and only if $p$ lies in the $i$-th row and $j$-th column of $\Psi_{r^*}(P)$. This can be done by applying the algorithm in Lemma~\ref{lem:20} but replacing the CS algorithm with the WX algorithm as the decision algorithm. The runtime becomes $O(n\log^3n)$ because the WX algorithm runs in $O(n\log^2 n)$ time. 

Let $(r_1, r_2]$ denote the interval after building the grid. We pick any $r \in (r_1, r_2)$ and compute the grid information of $\Psi_r(P)$, which has the same combinatorial structure as $\Psi_{r^*}(P)$ if $r^*\neq r_2$.  Below, we will simply use $\Psi(P)$ to refer to the grid information computed above, meaning that it does not change with respect to $r\in (r_1,r_2)$.



We use $dist_r[\cdot]$, $Q(r)$, $z(r)$ respectively to refer to $dist[\cdot]$, $Q$, $z$ in the WX algorithm running on a parameter $r$.
We start with setting $dist_r[s] = 0$, $dist_r[p] = \infty$ for all $p \in P\setminus\{s\}$, and $Q(r)=P$. 

Next we enter the main loop. As long as $Q(r)\neq \emptyset$, in each iteration, we will find a point $z(r)$ with the minimum $dist_r$-value from $Q(r)$ and update $dist_r$-values for points in $Q(r)_{\square_{z(r)}} \cup Q(r)_{\boxplus_{z(r)}}$. Points in $Q(r)_{\square_{z(r)}}$ are then removed from $Q(r)$. Each iteration will shrink $(r_1,r_2]$ such that the following algorithm invariant is maintained: $(r_1,r_2]$ contains $r^*$ and if $r^*\neq r_2$, the following holds for all $r\in (r_1,r_2)$: $z(r)=z(r^*)$, $Q(r)=Q(r^*)$, and $dist_r[p] = dist_{r^*}[p]$ for all $p\in P$.

Consider an iteration of the main loop. We assume that the invariant holds before the iteration on the interval $(r_1,r_2]$, which is true before the first iteration. In the following, we describe our algorithm for the iteration and we will show that the invariant holds after the iteration. We assume that $r^*\neq r_2$. According to our invariant, for any $r\in (r_1,r_2)$, we have $z(r)=z(r^*)$, $Q(r)=Q(r^*)$, and $dist_r[p] = dist_{r^*}[p]$ for all $p\in P$.

We first find a point $z(r) \in Q(r)$ with the minimum $dist_r$-value. Since the invariant holds before the iteration, we have $z(r) = \arg \min_{p \in Q(r)} dist_r[p] = \arg \min_{p \in Q(r^*)} dist_{r^*}[p] = z(r^*)$.\footnote{When picking $z(r)$, we break ties following the same way as the WX algorithm. This guarantees $z(r)=z(r^*)$ even if ties happen.} Hence, no ``parameterization'' is needed in this step, i.e., all involved values in the computation of this step are independent of $r$.

Next, we perform the first update \textsc{Update}($Q(r)_{\boxplus_{z(r)}}$, $Q(r)_{\square_{z(r)}}$). This step also does not need parameterization. Indeed, for each point $p\in Q(r)_{\boxplus_{z(r)}}$, we assign $dist_r[p]$ to $p$ as a weight, and then construct the additively-weighted Voronoi diagram on $Q(r)_{\boxplus_{z(r)}}$. For each point $v \in Q(r)_{\square_{z(r)}}$, we use the diagram to find its additively-weighted nearest neighbor $q_v(r) \in Q(r)_{\boxplus_{z(r)}}$ and update $dist_r[v] = \min \{dist_r[v], dist_r[q_v(r)] + \lVert q_v(r) - v \rVert\}$. Since $z(r)=z(r^*)$, and $Q(r)=Q(r^*)$, we have $Q(r)_{\boxplus_{z(r)}} = Q(r^*)_{\boxplus_{z(r^*)}}$ and $Q(r)_{\square_{z(r)}} = Q(r^*)_{\square_{z(r^*)}}$. Further, since $dist_r[p] = dist_{r^*}[p]$ for all $p \in P$,
for each point $v \in Q(r)_{\square_{z(r)}}$, $q_v(r)=q_v(r^*)$ and
each updated $dist_r[v]$ in our algorithm is equal to the corresponding updated $dist_{r^*}[v]$ in the same iteration of the WX algorithm running on $r^*$. As such, the invariant still holds after the first update.

Implementing the second update \textsc{Update}($Q(r)_{\square_{z(r)}}$, $Q(r)_{\boxplus_{z(r)}}$) is more challenging and parameterization is necessary. It suffices to implement \textsc{Update}($Q(r)_{\square_{z(r)}}$, $Q(r)_{\square}$) for all cells $\square\in \boxplus_{z(r)}$.

If $\square$ is $\square_{z(r)}$, then $Q(r)_{\square_{z(r)}} = Q(r)_{\square}$. In this case, again no parameterization is needed.
Since the distance between any two points in $\square_{z(r)}$ is at most $r$, we can easily implement \textsc{Update}$(Q(r)_{\square_{z(r)}}, Q(r)_{\square})$ in $O(|Q(r)_{\square_z(r)}| \log |Q(r)_{\square_z(r)}|)$ time, by first building a additively-weighted Voronoi diagram on points of $Q(r)_{\square_{z(r)}}$ (each point $p\in Q(r)_{\square_{z(r)}}$ is assigned a weight equal to $dist_r[p]$), and then using it to find the additively-weighted nearest neighbor $q_v(r)$ for each point $v\in Q(r)_{\square_z}$.
By an analysis similar to the above first update, the invariant still holds.

We now consider the case where $\square$ is not $\square_{z(r)}$. In this case, $\square$ and $\square_{z(r)}$ are separated by an axis-parallel line $\ell$. Without loss of generality, we assume that $\ell$ is horizontal and $\square_{z(r)}$ is below $\ell$.
Since $z(r)=z(r^*)$ and $Q(r)=Q(r^*)$ for all $r\in (r_1,r_2)$, we let $U=Q(r)_{\square_{z(r)}}$ and $V=Q(r)_{\square}$, meaning that both $U$ and $V$ are independent of $r\in (r_1,r_2)$.
Recall that there are three steps in the second update of the decision algorithm. Our algorithm needs to simulate all three steps. As will be seen later, only the second step needs parameterization.

The first step is to sort points in $U$ by their $dist_{r}$-values. Since $dist_r[p] = dist_{r^*}[p]$ for all $p \in P$, the sorted list $\{u_1, u_2, ..., u_{|U|}\}$ of $U$ obtained in our algorithm is the same as the sorted list obtained in the decision algorithm running on $r^*$.

For any $r$, we use $\bigodot_p(r)$ to denote the disk centered at a point $p$ with radius $r$.

The second step is to compute $|U|$ disjoint subsets $\{V_1(r), V_2(r), ..., V_{|U|}(r)\}$ of $V$ such that $V_i(r) = \{v\ |\ i_v(r) = i, v\in V\}$, where $i_v(r)$ is the smallest index such that $\bigodot_{u_{i_v(r)}}(r)$ contains point $v$.
This step needs parameterization. We will shrink the interval $(r_1, r_2]$ so that it still contains $r^*$ and if $r^* \neq r_2$, then for any $r \in (r_1, r_2)$, $V_i(r) = V_i(r^*)$ holds for all $1 \leq i \leq |U|$ (it suffices to ensure $i_v(r) = i_v(r^*)$ for all $v \in V$).
Our algorithm relies on the following observation, which is based on the definition of $i_v(r)$.
\begin{observation}\label{obser:10}
For any point $v \in V$, if $\bigodot_{u_j}(r)$ contains $v$ with $1 \leq j \leq |U|$, then $i_v(r) \leq j$.
\end{observation}

For a subset $P'\subseteq P$, let $\calF_r(P')$ denote the union of
the disks centered at points of $P'$ with radius $r$.
We first solve a subproblem in the following lemma.


\begin{lemma}
    \label{lem:subproblem}
    Suppose $(r_1,r_2]$ contains $r^*$ such that if $r^*\neq r_2$,
	then for all $r\in (r_1,r_2)$, $dist_r[p]=dist_{r^*}[p]$ for all
	points $p\in P$. For a subset $U'\subseteq U$ and a subset
	$V'\subseteq V$, in $O(n \log^2 n \cdot
	\log (|U'| + |V'|))$ time we can shrink $(r_1,r_2]$ so that it
	still contains $r^*$ and if $r^*\neq
	r_2$, then for all $r\in (r_1,r_2)$, for any $v\in V'$, $v$ is contained in
	$\calF_{r}(U')$ if and only if $v$ is contained in
	$\calF_{r^*}(U')$.
\end{lemma}
\begin{proof}
    \label{proof:lem-subproblem}
    Recall that all points of $U$ are below $\ell$ and all points of $V$ are above $\ell$. 
    For any $r$, the problem to determine whether $v$ is contained in $\calF_r(U')$ for each $v\in V'$ is an instance of Subproblem~\ref{subpro:10} (i.e., consider the points of $U'$ as red points and the points of $V'$ as blue points). Recall that solving Subproblem~\ref{subpro:10} for a fixed $r$ involves three subroutines and we also give a parameterized algorithm for solving it on the unknown $r^*$ in Section~\ref{subsec:bfs} for the unweighted case. Here, to achieve the lemma, we can essentially apply the same algorithm as in Section~\ref{subsec:bfs} but instead use the WX algorithm as the decision algorithm. We sketch it below.
    
    Let $\calU_r(U')$ denote the upper envelope of the portions of the disks $\bigodot_u(r)$ above $\ell$ for all $u\in U'$.
	A point $v\in V'$ is in $\calF_{r^*}(U')$ if and only if
	$v$ is below $\calU_r(U')$.
    The algorithm has three subroutines. 
     The first subroutine is to shrink $(r_1,r_2]$ so that it still	contains $r^*$ and if $r^*\neq r_2$, then for all $r\in (r_1,r_2)$, $\calU_r(U')$ has the same combinatorial structure as $\calU_{r^*}(U')$.
     This can be done by applying the algorithm of Section~\ref{subsubsec:ComputingTheUpperEnvelope} but using the WX algorithm as the decision algorithm.      
The second subroutine is to shrink $(r_1, r_2]$ such that it still contains $r^*$ and if $r^* \neq r_2$, then for all $r \in (r_1, r_2)$, the sorted list of the vertices of $\calU_r(U')$ and all points of $V'$ is the same as the sorted list of the vertices of $\calU_{r^*}(U')$ and all points of $V'$. This can be done by applying the algorithm of Section~\ref{subsubsec:ConnectivitiesBetweenBRPoints} but using the WX algorithm as the decision algorithm.
     The third subroutine is to shrink $(r_1, r_2]$ so that $(r_1, r_2]$ contains $r^*$ and if $r^* \neq r_2$, then for any $r \in (r_1, r_2)$, for any $v\in V'$, $v$ is below the arc spanning it in $\calU_r(U')$ if and only if $v$ is below the arc spanning it in $\calU_{r^*}(U')$. This can be done by applying the algorithm of Section~\ref{subsubsec:DecidingWhetherEachBluePoint} but using the WX algorithm as the decision algorithm.
     Following the analysis of Sections~\ref{subsubsec:ComputingTheUpperEnvelope},~\ref{subsubsec:ConnectivitiesBetweenBRPoints}, and \ref{subsubsec:DecidingWhetherEachBluePoint},
     the total time of the algorithm is bounded by $O(n \log^2 n \cdot \log (|U'| + |V'|))$ because the decision algorithm runs in $O(n\log^2 n)$ time (and both $|U'|$ and $|V'|$ are no more than $n$).
\end{proof}

Recall that we have an interval $(r_1, r_2]$. Our goal is to shrink it so that it still contains $r^*$ and if $r^* \neq r_2$, then for any $r \in (r_1, r_2)$, $V_i(r) = V_i(r^*)$ holds for all $1 \leq i \leq |U|$.
Based on Observation~\ref{obser:10} and using Lemma~\ref{lem:subproblem}, we have the following lemma.

\begin{lemma}
\label{lem:partition}
We can shrink the interval $(r_1, r_2]$ in $O(n \log^4 n)$ time so that it still contains $r^*$ and if $r^* \neq r_2$, then for any $r \in (r_1, r_2)$, $V_i(r) = V_i(r^*)$ holds for all $1 \leq i \leq |U|$.
\end{lemma}
\begin{proof}
\label{proof:lem-partition}
To have $V_i(r) = V_i(r^*)$ for all $1 \leq i \leq |U|$, it
suffices to ensure $i_v(r) = i_v(r^*)$ for all points $v \in V$.
Let $M=|U|$ and $N=|V|$. Note that $M\leq n$ and
$N\leq n$.

As defined in the proof of Lemma~\ref{lem:subproblem}, for any subset
$U' \subseteq U$ and any $r$, we use $\calU_r(U')$ to denote the upper
envelope of the portions of $\bigodot_u(r)$ above $\ell$ for all $u\in U'$.

In light of Observation~\ref{obser:10}, we use the divide and conquer approach.
Recall that $U=\{u_1,u_2,\ldots,u_M\}$. Consider the following subproblem on $(U,V)$: shrink
$(r_1, r_2]$ so that it still contains $r^*$ and if $r^* \neq r_2$, then for any $r \in (r_1,
r_2)$, for any $v\in V$, $v$ is below $\calU_r(U_1)$ if and only if $v$ is below $\calU_{r^*}(U_1)$, where
$U_{1}$ is the first half of $U$, i.e., $U_1= \{u_1, u_2,..., u_{\lfloor\frac{M}{2}\rfloor}\}$.
The subproblem can be solved in $O(n\log^3 n)$ time by applying Lemma~\ref{lem:subproblem}. Next, we pick any $r\in (r_1,r_2)$ and compute $\calU_r(U_1)$ and find the subset $V_1$ of the points of
$V$ that are below $\calU_r(U_1)$ (e.g., see
Fig.~\ref{fig:DivideAndConquer_1}). By Observation~\ref{obser:10}, for
each point $v\in V$, $i_v(r)\leq \lfloor\frac{M}{2}\rfloor$ if $v\in
V_1$ and $i_v(r)> \lfloor\frac{M}{2}\rfloor$ otherwise. By the above
property of $(r_1,r_2]$, for each point $v\in V$, we also have
$i_v(r^*)\leq \lfloor\frac{M}{2}\rfloor$ if $v\in V_1$ and $i_v(r^*)>
\lfloor\frac{M}{2}\rfloor$ otherwise.

\begin{figure}[t]
\begin{minipage}[t]{\linewidth}
\begin{center}
\includegraphics[totalheight=1.5in]{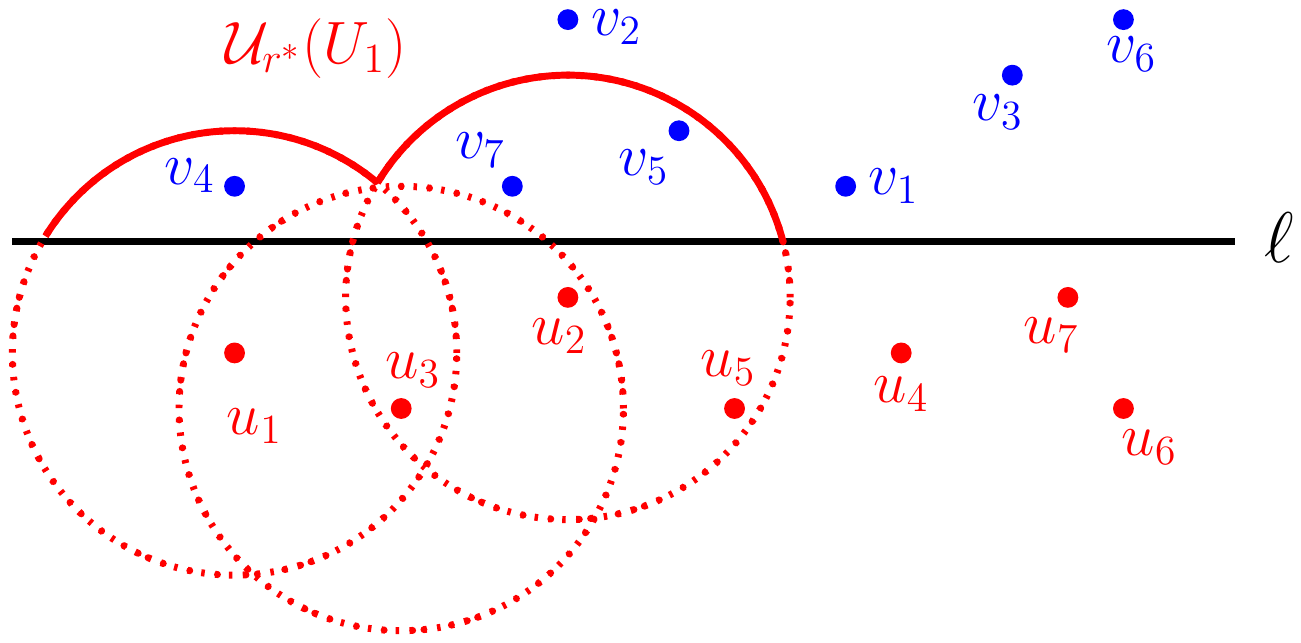}
\caption{Illustrating $U_1$ and $V_1$, where $U_1=\{u_1,u_2,u_3\}$ and $V_1=\{v_4,v_5,v_7\}$. The solid arcs are on $\calU_{r^*}(U_1)$. }\label{fig:DivideAndConquer_1}
\end{center}
\end{minipage}
\end{figure}

Next, we solve two subproblems recursively: one on $(U_1,V_1)$ and the other on $(U\setminus U_1, V\setminus V_1)$. Both subproblems use $(r_1,r_2]$ as their ``input intervals'' and solving each subproblem will produce a shrunk ``output interval'' $(r_1,r_2]$. Consider a subproblem on $(U',V')$ with $U'\subseteq U$ and $V'\subseteq V$. If $|U'|=1$, then we solve the problem ``directly'' (i.e., this is the base case) as follows. Assume that $r^*\neq r_2$ and let $r$ be any value in $(r_1,r_2)$. Let $u_j$ be the only point of $U'$. If $j<M$, according to our algorithm and based on Observation~\ref{obser:10}, $i_v(r)=i_v(r^*)=j$ holds for all points $v\in V'$. If $j=M$, however, for each point $v\in V'$, it is possible that $v$ is not contained in $\bigodot_{u}(r^*)$ for any point $u\in U$, in which case $v$ is not below $\calU_{r^*}(U)$ and thus is not below $\calU_{r^*}(U')$. On the other hand, if $v$ is below $\calU_{r^*}(U')$, then $i_v(r^*)=M$.
To solve the problem, we can simply apply Lemma~\ref{lem:subproblem} on $U'$ and $V'$, after which we obtain an interval $(r_1,r_2]$. Then, we pick any $r\in (r_1,r_2)$ and for any $v\in V'$ with $v$ contained in $\bigodot_{u_M}(r)$, $i_v(r)=i_v(r^*)=M$ holds if $r^*\neq r_2$.

The above divide-and-conquer algorithm can be viewed as a binary tree structure $T$ in which each node represents a subproblem. Clearly, the height of $T$ is $O(\log M)$ and $T$ has $\Theta(M)$ nodes.
If we solve each subproblem individually by Lemma~\ref{lem:subproblem} as described above, then the algorithm would take $\Omega(Mn)$ time because there are $\Omega(M)$ subproblems and solving each subproblem by Lemma~\ref{lem:subproblem} takes $\Omega(n)$ time, which would result in an $\Omega(n^2)$ time algorithm in the worst case. To reduce the runtime, instead, we solve subproblems at the same level of $T$ simultaneously (or ``in parallel'') by applying the algorithm of Lemma~\ref{lem:subproblem}, as follows.

Consider all subproblems in the same level of $T$; let $S$ denote the set of all these subproblems. There is an input interval $(r_1,r_2]$ for all subproblems of $S$, which is true initially at the root for $(U,V)$. After solving all subproblems in this level, our algorithm will produce a single shrunk interval $(r_1,r_2]$, which will be used as the input interval for all subproblems in the next level of $T$.

Recall that the algorithm of Lemma~\ref{lem:subproblem} has three subroutines (which follow the algorithm in Section~\ref{subsec:bfs}), each of which involves computing a set of critical values and then performing binary search on them using the decision algorithm to shrink the interval $(r_1,r_2]$. To solve all subproblems of $S$ simultaneously using the algorithm of Lemma~\ref{lem:subproblem}, our idea is that in each of the three subroutines, we perform binary search on the critical values of all subproblems of $S$ (this again follows the same way as in Section~\ref{subsec:bfs}, where critical values of all instances of $\calI$ are considered altogether), i.e., we solve all these subproblems ``in parallel''. In this way, solving all subproblems of $S$ together only needs to call the decision algorithm $O(\log n)$ times. The details are given below.

For the first subroutine, the goal is to determine the combinatorial structure of the upper envelope. For each subproblem on $(U',V')$, we compute the Voronoi diagram for $U'$ and then find the critical values. Notice that the subsets $U'$ (resp., $V'$) for all subproblems of $S$ form a partition of $U$ (resp., $V$), and thus the total time for building the diagram and computing the critical values for all subproblems of $S$ takes $O((M+N)\log (M+N))$ time in total. Also, the total number of critical values is $O(N)$. Performing the binary search on these critical values as before can be done in $O(n\log^2 n \cdot \log N)$ time, after which we obtain a shrunk interval $(r_1,r_2]$. This finishes the first subroutine for all subproblems of $S$, which takes $O(n\log^3 n)$ time (since $M\leq n$ and $N\leq n$).

The second subroutine is to sort all points of $V'$ in each subproblem on $(U',V')$ along with the vertices of the upper envelope $\calU_{r^*}(U')$. We now put all involved points of all subproblems of $S$ in one coordinate system and sort them altogether (in the same way as in Section~\ref{subsubsec:ConnectivitiesBetweenBRPoints}). Since the subsets $V'$ (resp., $U'$) of all subproblems of $S$ form a partition of $V$ (resp., $U$), the total number of points in the subsets $V'$ in all subproblems of $S$ is $N$. Also, the number of vertices of $\calU_{r^*}(U')$ is proportional to $|U'|$. Hence, the total number of vertices of the upper envelopes $\calU_{r^*}(U')$ in all subproblems of $S$ is $O(M)$. As such, the total number of points we need to sort is $O(M+N)$. We apply the same algorithm as before to sort them, i.e., Cole's parametric search~\cite{ref:ColeSl87} with AKS sorting network~\cite{ref:AjtaiPr83} and our decision algorithm. 
Sorting all involved points can be done in $O(n\log^2 n \cdot \log (M+N))$ time, after which a shrunk interval $(r_1,r_2]$ is obtained. This finishes the second subroutine for all subproblems of $S$, which takes $O(n\log^3 n)$ time.


For the third subroutine, we collect the critical values in each subproblem of $S$ in the same way as before. The total number of critical values for all subproblems is $N$. We perform binary search on these critical values in the same way as before, after which a shrunk interval $(r_1,r_2]$ is obtained. The total time is $O(n\log^2 n \cdot \log N)$. This finishes the third subroutine for all subproblems, which takes $O(n\log^3 n)$ time. The final interval $(r_1,r_2]$ will be used as the input interval for all subproblems in the next level of $T$.

In summary, solving all subproblems in the same level of $T$ can be done in $O(n\log^3n)$ time. As $T$ has $O(\log M)$ levels, the total time of the overall algorithm is $O(n\log^4n)$.
\end{proof}

With Lemma~\ref{lem:partition}, we obtain subsets $\{V_1(r), V_2(r), ..., V_{|U|}(r)\}$ and an interval $(r_1, r_2]$ containing $r^*$ such that if $r^* \neq r_2$, for any $r \in (r_1, r_2)$, $V_i(r) = V_i(r^*)$ holds for all $1\leq i\leq |U|$. Note that neither the array $dist_r[\cdot]$ nor $Q(r)$ is modified during the algorithm of Lemma~\ref{lem:partition}. Hence, if $r^*\neq r_2$, for all $r\in (r_1,r_2]$, we still have $Q(r)=Q(r^*)$ and $dist_r[p]=dist_{r^*}[p]$ for all points $p\in P$. Thus, our algorithm invariant still holds. This finishes the second step of the second update.

The third step of the second update is to solve the offline insertion-only additively-weighted nearest neighbor searching problem. This step does not need parameterization. Similar to the first update, we pick any $r \in (r_1, r_2)$ and apply the WX algorithm directly. Indeed, the algorithm on $r^*$ only relies on the following information: $U$ and its sorted list by $dist_{r^*}[\cdot]$ values and the subsets $V_1(r^*),\ldots,V_{|U|}(r^*)$. Recall that if $r^*\neq r_2$, then for all $r\in (r_1,r_2)$, $dist_r[p]=dist_{r^*}[p]$ for all $p\in P$, and $V_i(r)=V_i(r^*)$ for all $1\leq i\leq |U|$. As such, if we pick any $r \in (r_1, r_2)$ and apply the WX algorithm directly,  $dist_r[v]=dist_{r^*}[v]$ holds for all points $v\in V$ after this step. Therefore, as in the WX algorithm, this step can be done in $O(k\log^2 k)$ time, where $k=|U|+|V|$.

This finishes the second update of the algorithm. As discussed above, the algorithm invariant holds for the interval $(r_1,r_2]$.


The final step of the iteration is to remove points in $Q(r)_{\square_{z(r)}}$ from $Q(r)$. Since if $r^* \neq r_2$, for all $r \in (r_1, r_2)$, $Q(r)=Q(r^*)$, $z(r) = z(r^*)$, and $Q(r)_{\square_{z(r)}}=Q(r^*)_{\square_{z(r^*)}}$,
$Q(r)=Q(r^*)$ still holds after this point removal operation. Therefore, our algorithm invariant holds after the iteration.

In summary, each iteration of our algorithm takes $O(n \log^4 n)$ time. If the point $t$ is contained in $\square_{z(r)}$ (i.e., $t$ is reached) in the current iteration, then we terminate the algorithm. The following lemma shows that we can simply return $r_2$ as $r^*$.

\begin{lemma}
\label{lem:correctness}
     Suppose that $t$ is contained in $\square_{z(r)}$ in an iteration of our algorithm and $(r_1,r_2]$ is the interval after the iteration. Then $r^* = r_2$.
\end{lemma}
\begin{proof}
\label{proof:lem-correctness}
Assume to the contrary that $r^* \neq r_2$. Then we have $r^* \in (r_1, r_2)$ since $r^* \in (r_1, r_2]$. Let $r' = (r_1 + r^*) / 2$, and thus $r' \in (r_1, r_2)$ and $r' < r^*$. By our algorithm invariant and the correctness of the WX algorithm ($dist_{r}[p] = d_{r}(s, p)$ for all points $p \in P_{\square_{z(r)}}$ after the iteration), we have $d_{r'}(s, t) = dist_{r'}[t] = dist_{r^*}[t] = d_{r^*}(s, t)$. By the definition of $r^*$, $d_{r^*}(s, t) \leq \lambda$. Therefore, $d_{r'}(s, t) \leq \lambda$. But this contradicts with the definition of $r^*$ since $r^*=\arg \min_{r} \{ d_{r}(s, t) \leq \lambda \}$. The lemma thus holds.
\end{proof}

The algorithm may take $\Omega(n^2)$ time because $t$ may be reached in $\Omega(n)$ iterations. A further improvement is discussed in the next subsection.

\subsection{A further improvement}

To further reduce the runtime of the algorithm, we borrow a technique from~Section~\ref{sec:second} to partition the cells of the grid into large and small cells.

As before, we first compute the grid information $\Psi(P)$ and obtain an interval $(r_1,r_2]$.
Let $\calC$ denote the set of all non-empty cells of $\Psi(P)$ (i.e., cells that contain at least one point of $P$). For each cell $C \in \calC$, let $N(C)$ denote the set of non-empty neighboring cells of $C$ in $\calC$ and $P(C)$ the set of points of $P$ contained in cell $C$. We have $|N(C)| = O(1)$ and $|\calC| = O(n)$. A cell $C$ of $\calC$ is a \emph{large cell} if it contains at least $n^{3/4}\log^{3/2}n$ points of $P$, i.e., $|P(C)| \geq n^{3/4}\log^{3/2}n$, and a \emph{small cell} otherwise. Clearly, $\calC$ has at most $n^{1/4}/\log^{3/2}n$ large cells. For all pairs of non-empty neighboring cells $(C, C')$, with $C\in \calC$ and $C' \in N(C)$, $(C, C')$ is a \emph{small-cell pair} if both $C$ and $C'$ are small cells, and a \emph{large-cell pair} otherwise, i.e., at least one cell is a large cell. Since $N(C)=O(1)$ for each cell $C\in \calC$, there are $O(n^{1/4}/\log^{3/2}n)$ large-cell pairs.

We follow the algorithmic framework in Section~\ref{sec:second}.
Notice that in each iteration of the main loop in our previous algorithm, only the second step of the second update parameterizes the WX algorithm (i.e., the decision algorithm is called on certain critical values); in that step, we need to process $O(1)$ pairs of cells $(C,C')$ with $C\in \calC$ and $C'\in N(C)$. No matter how many points of $P$ are contained in the two cells, we need $O(n\log^4 n)$ time to perform the parametric search due to Lemma~\ref{lem:partition}. To reduce the time, we preprocess all small-cell pairs so that the algorithm only needs to perform the parametric search for large-cell pairs. Since there are only $O(n^{1/4}/\log^{3/2}n)$ large-cell pairs, the total time we spend on parametric search can be reduced to $O(n^{5/4}\log^{5/2}n)$.
For those small-cell pairs, the preprocessing provides sufficient information to allow us to simply run the original WX algorithm without parametric search. Specifically, before we enter the main loop of the algorithm (and after the grid information $\Psi(P)$ is computed, along with an interval $(r_1,r_2]$), we preprocess all small-cell pairs using the following lemma.

\begin{lemma}
\label{lem:preprocess}
In $O(n^{5/4} \log^{5/2} n)$ time we can shrink the interval $(r_1, r_2]$
so that it still contains $r^*$ and if $r^* \neq r_2$, then for any $r \in (r_1, r_2)$, for any small-cell pair $(C, C')$ with $C \in \calC$ and $C' \in N(C)$, an edge connects a point $p \in P(C)$ and a point $p' \in P(C')$ in $G_r(P)$ if and only if an edge connects $p$ and $p'$ in $G_{r^*}(P)$.
\end{lemma}
\begin{proof}
\label{proof:lem-preprocess}
Lemma~\ref{lem:40} essentially solves the same problem for the unweighted case. 
Here we follow the same algorithm as in Lemma~\ref{lem:40} but replace their decision algorithm
by our decision algorithm for the weighted case.
The algorithm has $O(\log n)$ iterations, and following the same analysis as in Lemma~\ref{lem:40} and using the new threshold $n^{3/4}\log^{3/2}n$ for defining large cells, one can show that each iteration takes $O(n^{5/4}\log^{3/2} n)$ time. More specifically, if we use the same notation as in the proof of Lemma~\ref{lem:40}, then we have $n_i \leq 2 \cdot n^{3/4}\log^{3/2}n$, and thus $|\calL|=\sum_{i = 1}^{m} n_i^{4/3} \log n_i$ is bounded by $O(n^{5/4} \log^{3/2} n)$. 
Therefore, the total running time of the algorithm is $O(n^{5/4}\log^{5/2} n)$.
\end{proof}

Let $(r_1, r_2]$ denote the interval obtained after the preprocessing for all small-cell pairs in Lemma~\ref{lem:preprocess}. Lemma~\ref{lem:preprocess} essentially guarantees that if $r^* \neq r_2$, then for any $r \in (r_1, r_2)$, the adjacency relation of points in any small-cell pair in $G_r(P)$ is the same as that in $G_{r^*}(P)$. Note that if $(r_1, r_2]$ is shrunk so that it still contains $r^*$, then the above property still holds for the shrunk interval. Based on this property, combining with our previous algorithm, we have the following theorem.

\begin{theorem}
\label{theorem:weightedL2}
    The reverse shortest path problem for $L_2$ weighted unit-disk graphs can be solved in $O(n^{5/4} \log^{5/2} n)$ time.
\end{theorem}
\begin{proof}
    \label{proof:theorem-weightedL2}
    The goal is to compute $r^*$.
    We first build a grid $\Psi(P)$ along with an interval $(r_1, r_2]$ in $O(n \log^3 n)$ time. Then we classify all non-empty cells in $\Psi(P)$ to large cells and small cells. Next, we use Lemma~\ref{lem:preprocess} to shrink the interval $(r_1, r_2]$ in $O(n^{5/4} \log^{5/2} n)$ time.

    We proceed to the main loop of the algorithm. In each iteration, we proceed in the same way as before except that the second step of the second update \textsc{Update}($Q(r)_{\square_{z(r)}}$, $Q(r)_{\boxplus_{z(r)}})$ is now executed as follows. Recall that it suffices to perform \textsc{Update}($Q(r)_C$, $Q(r)_{C'}$) with $C=\square_{z(r)}$ and $C'\in N(C)$. If $(C, C')$ is a large-cell pair, then we apply our parametric search procedure in the same way as before. Since the number of large-cell pairs is $O(n^{1/4}/\log^{3/2}n)$ and implementing the second step of  \textsc{Update}($Q(r)_C$, $Q(r)_{C'}$) with the parametric search takes $O(n \log^4 n)$ time by Lemma~\ref{lem:partition}. Thus the total time we spend on all large-cell pairs is $O(n^{5/4} \log^{5/2} n)$. If $(C, C')$ is a small-cell pair, according to the property of $(r_1,r_2]$ in the statement of
    Lemma~\ref{lem:preprocess}, we can simply pick any value $r \in (r_1, r_2)$ and then apply the WX algorithm directly. Following the time complexity of the WX algorithm, the second step of \textsc{Update}($Q(r)_C$, $Q(r)_{C'}$) of all small-cell pairs $(C,C')$ together takes $O(n \log n)$ time. The remaining parts of our algorithm together take the same running time as the WX algorithm, which is $O(n \log^2 n)$.

    We thus conclude that the total time of our overall algorithm is bounded by $O(n^{5/4} \log^{5/2} n)$.
\end{proof}

\section{The $L_1$ RSP Problem}
\label{sec:L1RSP}

In this section, we consider the $L_1$ RSP problem and present an $O(n \log^3 n)$ time algorithm for both the unweighted and the weighted cases. Unless otherwise stated, $G_r(P)$ and $d_r(u, v)$ are defined with respect to the $L_1$ metric in this section; also $\lVert p-q\rVert$ represents the $L_1$ distance between two points $p$ and $q$. Given a set $P$ of $n$ points, a value $\lambda$, and two points $s, t \in P$, the problem is to compute the smallest $r$ such that the shortest path length between $s$ and $t$ in the $L_1$ unit-disk graph $G_{r}(P)$ is at most $\lambda$. Let $r^*$ denote the optimal value $r$ for the problem. The goal is therefore to compute $r^*$.


Observe that $r^*$ must be equal to the $L_1$ distance of two points in $P$ in both the unweighted and the weighted cases. As already discussed in Section~\ref{sec:introduction}, the decision problem can be solved in $O(n\log n)$ time in both the unweighted case~\cite{ref:CabelloSh15,ref:ChanAl16} and the weighted case~\cite{ref:WangAn21}. In the following, we first discuss our algorithm for computing $r^*$ in the $L_1$ weighted case. As will be seen later, the $L_1$ unweighted case can be solved by exactly the same algorithm except that the decision algorithm is switched to that for the unweighted case. 
In the weighted case, the single-source-shortest-path algorithm of Wang and Zhao~\cite{ref:WangAn21} can be used to solve the decision problem in $O(n\log n)$ time; in the following, we use {\em decision algorithm} to refer to that algorithm.

Let $\Pi$ denote the set of the $L_1$ distances of all pairs of points of $P$. The main idea is to search $r^*$ in $\Pi$ by using the decision algorithm.
Our searching algorithm framework follows the $L_2$ distance selection algorithm in~\cite{ref:KatzAn97} but uses a different procedure to conduct the ``batched range searching'' to generate critical values. The algorithm has $O(\log n)$ stages and the $j$-th stage computes an interval $I_j = (a_j, b_j]$ (initially $I_0 = (0, \infty]$) containing $r^*$, such that $I_j\subseteq I_{j-1}$ and $|\Pi\cap I_j|=O(n^2\sigma^j)$ for some constant $0<\sigma<1$, i.e., the number of values of $\Pi$ in $I_j$ is a constant faction of the number of values of $\Pi$ in $I_{j-1}$. As such, after $O(\log n)$ stages, only a small amount of values of $\Pi$ remain, from which it is trivial to find $r^*$.
In the following, we describe the algorithm.


To simplify the discussion, we rotate the $x$- and $y$-axes by $45\degree$ and call them {\em new axes}; the original axes are referred to as the {\em old axes}. Correspondingly, each point in the plane has an {\em old coordinate} and a {\em new coordinate}.

We build a 2-dimensional range tree $T$ for $P$ following the new axes in $O(n \log^2 n)$ time~\cite{ref:deBergCo08}.
Specifically, following the sorted order of the points of $P$ by their new $x$-coordinates, we build a balanced binary search tree $T$ such that each leaf stores a point of $P$. For each node $v$ of $T$, let $P_v$ denote the set of points stored in the leaves of the subtree rooted at $v$. Following the sorted order of the points of $P_v$ by their new $y$-coordinates, we build a balanced binary search tree $T_v$ such that each leaf stores a point of $P_v$; for each node $u$ of $T_v$, let $P_u$ denote the subset of points stored in the leaves of the subtree of $T_v$ rooted at $u$ and we call $P_u$ a {\em canonical subset} of $T$.


For each point $p$ and a value $r$, we use $\diamondsuit_p(r)$ to denote the
$L_1$ disk centered at $p$ with radius $r$; note that $\diamondsuit_p(r)$ is a
diamond.

For each point $p \in P$, we intend to find the set $P_p(I)$ of points of $P$ whose
$L_1$ distances from $p$ lie in an interval $I = (a, b]$. Notice that
all these points must be in the {\em $L_1$ annulus} $A_p(I)$ that is the
region inside $\diamondsuit_{p}(b)$ and strictly outside $\diamondsuit_{p}(a)$. Further,
$A_p(I)$ can be decomposed into four rectangles whose
edges are parallel with the new axes (e.g., see
Fig.~\ref{fig:OrthogonalRangeSearching}). Hence, points of $P_p(I)$
can be found by performing four orthogonal range queries on $T$. Each query, which takes
$O(\log^2 n)$ time, returns $O(\log^2 n)$ pairwise-disjoint canonical subsets of $T$
whose union is the subset of points of $P$ in the query rectangle~\cite{ref:deBergCo08}.
Hence, $P_p(I)$ can be obtained in $O(\log^2 n)$ time
as the union of $O(\log^2 n)$ pairwise-disjoint canonical subsets of
$T$.

\begin{figure}[t]
    \centering
    \begin{minipage}{\textwidth}
        \centering
        \includegraphics[height=1.8in]{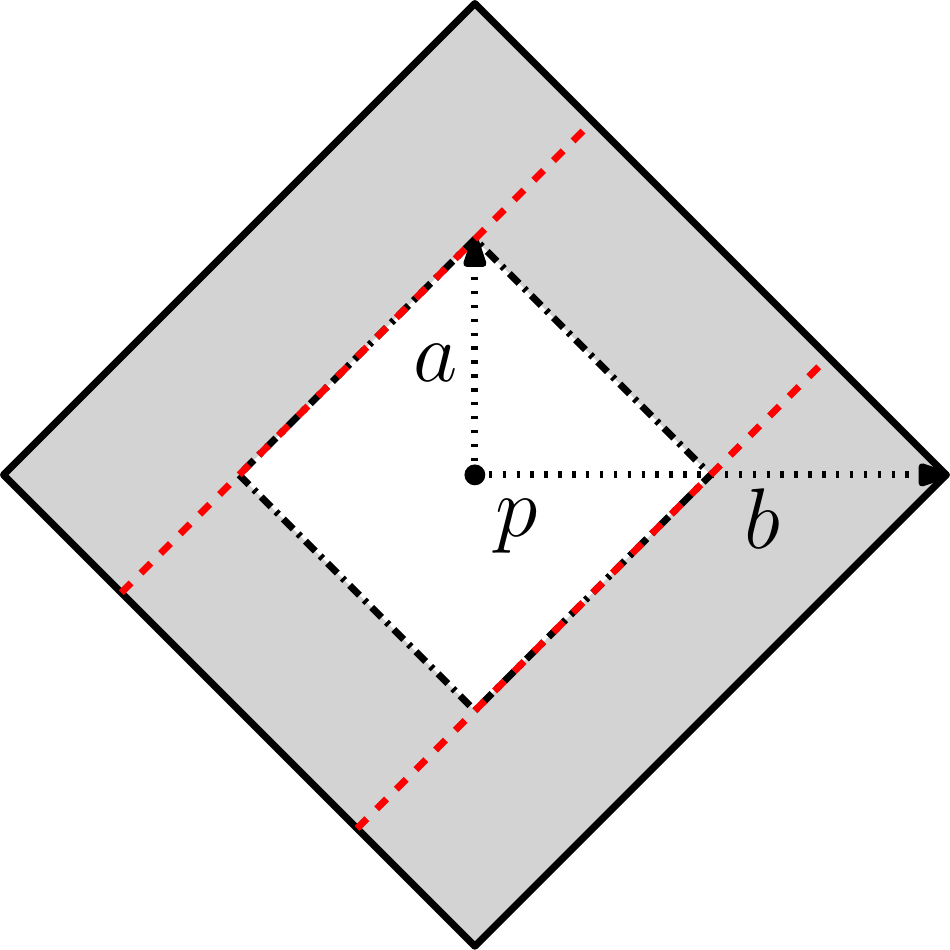}
        \caption{\footnotesize The grey region is the $L_1$ annulus $A_p(I)$ with $I=(a,b]$. The two dashed (red) segments decompose $A_p(I)$ into four rectangles.}
        \label{fig:OrthogonalRangeSearching}
    \end{minipage}
\end{figure}


Assume that we have an interval $I_{j - 1} = (a_{j - 1}, b_{j - 1}]$
(initially $j=1$ and $I_0=(0,\infty]$) such that the number of values of $\Pi$ in $I_{j-1}$ is $O(n^2 \sigma^{j-1})$ for some constant $\sigma\in (0,1)$ to be specified later, which is true initially when $j=1$.
The $j$-th stage of our
algorithm works as follows.
\begin{enumerate}
\item For each point $p \in P$, using the range tree $T$, we find the
collection $\calR_p$ of $O(\log^2 n)$ canonical subsets of $T$ whose union is
$P_p(I_{j-1})$. Computing $\calR_p$ for all $p\in P$ takes $O(n\log^2 n)$ time. Also, $\sum_{p\in P}|\calR_p|=O(n\log^2 n)$.

Note that $T$ has $O(n\log n)$ canonical subsets, denoted by $\calL_g$, $g=1,2,\ldots,O(n\log n)$, and their total size is $O(n\log^2 n)$~\cite{ref:deBergCo08}.
For each $\calL_g$, define $\calK_g=\{p\in P\ |\ \calL_g\in \calR_p\}$. As
$\sum_{p\in P}|\calR_p|=O(n\log^2 n)$, we have
$\sum_g|\calK_g|=O(n\log^2n)$ and constructing all $\calK_g$'s can be
done in $O(n\log^2 n)$ time by enumerating $\calR_p$ for all $p\in
P$.

The values of $\Pi$ in $I_{j-1}$ can be represented as the above collection of $O(n \log n)$
pairs $(\calK_g, \calL_g)$. Indeed, for any point $p\in \calK_g$ and
any point $q\in \calL_g$, we have $\lVert p-q\rVert\in I_{j-1}$. On
the other hand, for any two points $p,q\in P$ with $\lVert
p-q\rVert\in I_{j-1}$, there is a unique $g$ such that $p\in \calK_g$
and $q\in \calL_g$ and another unique $g$ such that $q\in \calK_g$ and
$p\in \calL_g$. Therefore, $\sum_g |\calK_g| \cdot |\calL_g|$ is twice
the number of values of $\Pi\cap I_{j - 1}$, and thus $\sum_g
|\calK_g| \cdot |\calL_g|=O(n^2 \sigma^{j-1})$.



\item Each pair $(\calK_g, \calL_g)$ can be viewed as a complete
bipartite graph $\calK_g \times \calL_g$ in the sense that the
distance between any point $p \in \calK_g$ and any point $q \in
\calL_g$ lies in $I_{j - 1}$.
For each $g$, we partition $\calK_g$ into subsets each of which has a size similar to that of $\calL_g$, as follows.
Let $m_g = |\calK_g|$ and $n_g = |\calL_g|$.

If $m_g \geq n_g$, then we partition $\calK_g$ into
$k = \lfloor \frac{m_g}{n_g} \rfloor$ subsets, $\calK_{g1},
\calK_{g2}, ..., \calK_{gk}$, where each subset contains $n_g$
elements except the last subset $\calK_{gk}$ contains at least
$n_g$ but at most $2n_g - 1$ elements.

If $m_g <n_g$, we exchange the names of $\calK_g$
and $\calL_g$, i.e., we use $\calK_g$ to refer to $\calL_g$ and use
$\calL_g$ to refer to the original $\calK_g$. Similarly, $m_g$ now
refers to the size of the new $\calK_g$ and $n_g$ the size
of the new $\calL_g$.
Then, we have $m_g>n_g$ and perform the same partition of (the new)
$\calK_g$ as above.
Note that the name exchange is only for convenience of discussions; alternatively, one could partition $\calL_g$ instead without doing the name exchange (but the discussion would become more tedious).
As will be seen, the name exchange will not affect the correctness of our
algorithm (i.e., Lemma~\ref{lem:L1case}), which only relies on the following properties $\sum_g
|\calK_g|=O(n\log^2 n)$, $\sum_g |\calL_g| = O(n \log^2 n)$, and $\sum_g |\calK_g|
\cdot |\calL_g| = O(n^2 \sigma^{j - 1})$. It is not difficult to see that the
name exchange does not affect these properties.

Next, for each $1 \leq i \leq k$, we consider $(\calK_{gi},\calL_g)$ as a complete bipartite graph and
construct a {\em $d$-regular LPS-expander graph} $E_{gi}$ on the vertex set $\calK_{gi} \cup
\calL_g$~\cite{ref:KatzAn97, ref:LubotzkyEx86}, for a fixed
constant $d$. The graph $E_{gi}$, which has $O(|\calK_{gi}| + |\calL_g|)$ edges,
can be computed in $O(|\calK_{gi}| + |\calL_g|)$ time~\cite{ref:KatzAn97, ref:LubotzkyEx86}.
Let $E_g$ be the union of these expander graphs, i.e.,
$E_{g} = \cup_{i = 1}^{k} E_{gi}$. Hence, the total time for
constructing $E_g$ is on the order of $\sum_{i = 1}^{k} (|\calK_{gi}| + |\calL_g|)$, which is bounded by
$O(|\calK_g| + \lfloor \frac{m_g}{n_g} \rfloor \cdot |\calL_g|) =
O(|\calK_g|)$. Therefore, constructing $E_g$'s for all $g$ takes
$\sum_g O(|\calK_g|) = O(n \log^2 n)$ time. Also, the number of
edges in $E_g$ is $O(|\calK_g| + |\calL_g|)$, and thus the total number
of edges of all $E_g$'s is $\sum_g O(|\calK_g| + |\calL_g|) = O(n
\log^2 n)$.

\item
Each edge connecting $p$ and $q$ in each $E_g$ with $p \in \calK_g$ and
$q\in \calL_g$ is associated with the $L_1$ distance $\lVert p-q\rVert$. Let
$W$ denote the set of all these distances over all $\{E_g\}_g$. Then
the size of $W$ is $O(n\log^2 n)$. Note that all values of $W$
are in $I_{j-1}$.

\item
Using the decision algorithm, we perform binary search on $W$ to find
the largest value $a_j\in W$ with $a_j<r^*$ and the smallest value
$b_j\in W$ with $r^*\leq b_j$. Define $I_j=(a_j,b_j]$. Hence,
$I_j\subseteq I_{j-1}$ and $I_j$ contains $r^*$. Also, $(a_j, b_j)$
does not contain any value of $W$. Note that when performing binary
search, we do not have to sort $W$ first, but instead use the linear time
selection algorithm~\cite{ref:BlumTi73}. As such,
since $|W|=O(n\log^2 n)$ and the decision algorithm takes $O(n\log n)$ time, finding $a_j$
and $b_j$ in $W$ can be done in $O(n \log^2 n)$ time.
\end{enumerate}

This finishes the $j$-th stage of the algorithm, which runs in $O(n\log^2 n)$ time.
Since $|\Pi\cap I_{j-1}|=O(n^2\sigma^{j-1})$, the following lemma
proves that $|\Pi\cap I_j|=O(n^2\sigma^j)$.

\begin{lemma}
    \label{lem:L1case}
    There exists a constant $\sigma\in (0,1)$ such that the number of values of $\Pi$ in $I_j = (a_j, b_j]$ is at most $\sigma$ times the number of values of $\Pi$ in $I_{j - 1} = (a_{j - 1}, b_{j - 1}]$.
\end{lemma}
\begin{proof}
    \label{proof:lem-L1case}
    Consider a pair $(\calK_{gi}, \calL_g)$ obtained from Step 2 in the $j$-th stage of our algorithm.
    Let $\calA_{gi}$ be the set of the $L_1$ annuli $A_p(I_j)$ of all points $p\in \calK_{gi}$. Let $H_{gi}$ be the set of supporting lines of all edges of the $L_1$ annuli of $\calA_{gi}$. For a parameter $h\leq |H_{gi}|$ to be specified later, we consider a $1/h$-cutting $\Xi$ of size $O(h^2)$ for the lines of $H_{gi}$, which consists of $O(h^2)$ (possibly unbounded) triangles (called {\em cells}) such that each cell is crossed by at most $|H_{gi}|/h$ lines of $H_{gi}$ (implying that the number of annuli of $\calA_{gi}$ intersecting each cell is $O(|\calA_{gi}|/h)$). Note that such a cutting always exists~\cite{ref:ChazelleA90,ref:ChazelleCu93}. Our algorithm does not need to compute the cutting and we use it here for the analysis purpose only.


    For ease of exposition, we assume that each point of $\calL_g$ is contained in the interior of a cell of $\Xi$.
    For each cell $\Delta\in \Xi$, let $\calL_g(\Delta)$ denote the set of points of $\calL_g$ inside $\Delta$, $\calA_{gi}'(\Delta)$ the set of annuli in $\calA_{gi}$ that fully contain $\Delta$, and $\calA_{gi}''(\Delta)$ the set of annuli in $\calA_{gi}$ that have an edge intersecting $\Delta$.
    Define $N_{gi}$ to be the number of $L_1$ distances between points of $\calK_{gi}$ and points of $\calL_g$ that lie in $I_j$. Then we have
    \begin{equation}\label{equ:10}
    N_{gi} \leq \sum_{\Delta\in \Xi} |\calA_{gi}'(\Delta)| \cdot |\calL_g(\Delta)| + \sum_{\Delta\in \Xi} |\calA_{gi}''(\Delta)| \cdot |\calL_g(\Delta)|.
    \end{equation}

    Since the number of annuli of $\calA_{gi}$ intersecting each cell $\Delta$ is $O(|\calA_{gi}|/h)$ and $|\calA_{gi}| = |\calK_{gi}|$, we have $|\calA_{gi}''(\Delta)| =O(|\calK_{gi}| / h)$. Note that $\sum_{\Delta\in \Xi} |\calL_g(\Delta)| = |\calL_g|$. As such, we obtain

    \begin{equation}\label{equ:20}
    \sum_{\Delta\in\Xi} |\calA_{gi}''(\Delta)| \cdot |\calL_g(\Delta)| = O\left(\frac{|\calK_{gi}| \cdot |\calL_g|}{h}\right).
    \end{equation}

    Consider the complete bipartite graph $(\calK_{gi},\calL_g)$. According to the definition of $a_j$ and $b_j$, the expander graph $E_{gi}$ has no edge whose associated $L_1$ distance lies in $(a_j,b_j)$. Note that for any $\Delta\in \Xi$, $\calL_g(\Delta)$ is a subset of $\calL_g$ and $\calK'_{gi}(\Delta)$ is a subset of $\calK_{gi}$, where $\calK'_{gi}(\Delta)$ is the set of centers of the annuli of $\calA'_{gi}(\Delta)$. For each point $p$ of $\calK'_{gi}(\Delta)$, since its corresponding annulus in $\calA'_{gi}(\Delta)$ fully contains $\Delta$, $\lVert p-q\rVert$ must lie in $(a_j,b_j)$ for any point $q\in \calL_g(\Delta)$. As such, the $L_1$ distance between any point $p\in \calK'_{gi}(\Delta)$ and any point $q\in \calL_g(\Delta)$ is in $(a_j,b_j)$.
    Hence, no edge in $E_{gi}$ connects a point of $\calL_g(\Delta)$ with a point of $\calK'_{gi}(\Delta)$.
    By Corollary 2.4 in~\cite{ref:KatzAn97}, if $A$ and $B$ are two vertex subsets of a $d$-regular expander graph of $N$ points and no edge of the graph connects a vertex of $A$ with a vertex of $B$, then $|A| \cdot |B| \leq 4 N^2 / d$. With this result, we obtain that $|\calA'_{gi}(\Delta)|\cdot |\calL_g(\Delta)|=|\calK'_{gi}(\Delta)|\cdot |\calL_g(\Delta)|\leq 4(|\calK_{gi}| + |\calL_g|)^2/d$ for any cell $\Delta\in \Xi$. As $\Xi$ has $O(h^2)$ cells, we can derive
    \begin{equation}\label{equ:30}
    \sum_{\Delta\in \Xi} |\calA'_{gi}(\Delta)| \cdot |\calL_g(\Delta)| \leq O(h^2) \cdot \frac{4(|\calK_{gi}| + |\calL_g|)^2}{d} = O\left(\frac{h^2 (|\calK_{gi}| + |\calL_g|)^2}{d}\right).
    \end{equation}

    Combining \eqref{equ:10}, \eqref{equ:20}, and \eqref{equ:30}, we have $$N_{gi} = O\left(\frac{h^2 (|\calK_{gi}| + |\calL_g|)^2}{d} + \frac{|\calK_{gi}| \cdot |\calL_g|}{h}\right).$$

    According to our partition of $\calK_g$ in Step 2 of each stage of the algorithm, it holds that $|\calL_g| \leq |\calK_{gi}| < 2 \cdot |\calL_g|$, which implies $(|\calK_{gi}| + |\calL_g|)^2 \leq 5 |\calK_{gi}| \cdot |\calL_g|$.
    Thus we have, $$N_{gi} = O\Big( [\frac{h^2}{d} + \frac{1}{h}] \cdot |\calK_{gi}| \cdot |\calL_g| \Big).$$

    By setting $h = d^{1/3}$ and $\sigma$ to be appropriately proportional to $1 / d^{1/3}$, we have $N_{gi} \leq \sigma \cdot |\calK_{gi}| \cdot |\calL_g|$. Summing up all these inequalities for all subsets $\calK_{gi}$ of $\calK_g$ leads to $N_g \leq \sigma \cdot |\calK_g| \cdot |\calL_g|$, where $N_g$ is the number of $L_1$ distances between points of $\calK_g$ and points of $\calL_g$ that lie in $I_j$.
    This further leads to $\sum_g N_g \leq \sigma \cdot \sum_g |\calK_g| \cdot |\calL_g|$.
    Note that $\sum_g N_g$ is equal to twice the number of values of $\Pi$ in $I_j$, i.e., $\sum_g N_g=2\cdot |\Pi\cap I_j|$, while $\sum_g |\calK_g| \cdot |\calL_g|$ is equal to twice the number of values of $\Pi$ in $I_{j-1}$, i.e., $\sum_g |\calK_g|=2\cdot |\Pi\cap I_{j-1}|$. Therefore, we obtain $|\Pi\cap I_j|\leq \sigma\cdot |\Pi\cap I_{j-1}|$. The lemma thus follows.
\end{proof}

By Lemma~\ref{lem:L1case}, after $O(\log n)$ stages, our algorithm will obtain an interval $I_j$ with  $|\Pi\cap I_j|\leq n$. Then, we can explicitly compute these values of $|\Pi\cap I_j|$ in $O(n\log^2 n)$ time using the range tree $T$ (as in Step 1 of each stage of our algorithm), after which $r^*$ can be easily found in $O(n\log^2 n)$ time by binary search on these values using the decision algorithm. As each stage runs in $O(n\log^2 n)$ time, the total time of the overall algorithm is $O(n \log^3 n)$.

For the unweighted case, we use exactly the same algorithm except that we switch to a decision algorithm for the unweighted case. Note that although the decision algorithm runs in $O(n)$ time after $O(n\log n)$ preprocessing for sorting~\cite{ref:ChanAl16}, each stage of the algorithm still takes $O(n\log^2 n)$ time. Hence, the total time is still $O(n \log^3 n)$. The following theorem summarizes our result.

\begin{theorem}
    \label{theorem:L1case}
    The reverse shortest path problem for $L_1$ unit-disk graphs in the unweighted/weighted case can be solved in $O(n \log^3 n)$ time.
\end{theorem}

\paragraph{The $L_1$ distance selection problem.} We remark that our technique can be
used to solve the $L_1$ distance selection problem in $O(n\log^3 n)$ time. Given a set $P$ of
$n$ points and an integer $k \in [1, \binom{n}{2}]$, the problem is to
find the $k$-th smallest value in $\Pi$, where $\Pi$ is the set of the
$L_1$ distances of all pairs of points of $P$. Katz and
Sharir~\cite{ref:KatzAn97} solved the $L_2$ version of the problem in
$O(n^{4/3} \log^2 n)$ time. Following their algorithmic scheme and
using our technique for the RSP problem (more specifically,
the technique for computing a compact representation for points of $P$ whose interpoint distances lie
in a given interval; similar techniques for a different problem were also given in~\cite{ref:KatzAn97}), we can solve the $L_1$ version of the problem in
$O(n\log^3 n)$ time. We briefly discuss it below. Suppose $r^*$ is the
$k$-th smallest value of $\Pi$ that we are looking for.

First of all, we need a decision algorithm for the decision problem:
Given any $r$, decide whether $r^*\leq r$. The following algorithm can
solve the decision problem in $O(n\log^2 n)$ time. First, we build a 2D range tree $T$ for $P$ as before. Then, for each point $p\in P$, we find the number $n_p$ of points of $P$ whose distances from $p$ are at most $r$, which can be done in $O(\log^2 n)$ time by an orthogonal range query on $T$. Observe that $1/2\cdot \sum_{p\in P}n_p$ is equal to the number of values of $\Pi$ smaller than or equal to $r$. Hence,
$r^*\leq r$ if and only if $1/2\cdot \sum_{p\in P}n_p\geq k$. Clearly, the total time of the algorithm is $O(n\log^2 n)$.

With the above decision algorithm, our algorithm for computing $r^*$ works as follows.
The algorithm again has $O(\log n)$ stages. In each stage, we compute an interval
$I_j = (a_j, b_j]$ and perform a binary search guided by the decision
algorithm on the distances of $\Pi$ in $I_j$. These
distances are represented by complete bipartite graphs in the same way as our RSP algorithm (i.e., the four steps). The difference is that we use the above decision algorithm in Step 4. Following the same analysis, we can still prove Lemma~\ref{lem:L1case}.
Because the new decision algorithm runs in $O(n\log^2 n)$ time, each stage now takes $O(n\log^3 n)$ time. The total time of the algorithm is thus $O(n \log^4 n)$.
A logarithmic factor can be further reduced using a ``Cole-like''
technique in exactly the same way as in~\cite{ref:KatzAn97}, so that the
number of calls of the decision algorithm in each stage can be
reduced to a constant. This improves the running time of overall algorithm to $O(n
\log^3 n)$.

One may wonder whether the ``Cole-like'' technique can improve the
runtime of our RSP algorithm in Theorem~\ref{theorem:L1case}.
Unfortunately this is not the case. Indeed, each stage of the RSP
algorithm runs in $O(n\log^2 n)$ time even if the time of
the decision algorithm is excluded. Hence, although the
``Cole-like'' technique can reduce the number of calls on the decision
algorithm, the total time of the algorithm is dominated by other parts of the algorithm, which is still $O(n\log^3 n)$.

\section{Concluding remarks}
\label{sec:con}

In this paper, we propose two algorithms for the RSP problem in $L_2$ unweighted unit-disk graphs with time complexities of $O(\lfloor \lambda \rfloor \cdot n\log n)$ and $O(n^{5/4}\log^{7/4} n)$, respectively. We also give an algorithm for the RSP problem in $L_2$ weighted unit-disk graphs with a time complexity of $O(n^{5/4}\log^{5/2} n)$. Interestingly, our second $L_2$ unweighted RSP algorithm and the $L_2$ weighted RSP algorithm break the $O(n^{4/3})$ time barrier for certain geometric problems~\cite{ref:EricksonOn95,ref:EricksonNe96}. In addition, we propose an algorithm that can solve the $L_1$ unweighted/weighted case in $O(n \log^3 n)$ time.

Our RSP problem is defined with respect to a pair of points $(s,t)$.
Our techniques can be extended to solve a more general ``single-source'' version of
the problem: Given a source point $s \in P$ and a value $\lambda$, compute the smallest value $r^*$ such that the lengths of shortest paths from $s$ to all vertices of $G_r(P)$ are at most $\lambda$, i.e., $\max_{t \in P} d_{r^*} (s, t) \leq \lambda$. The decision problem (i.e., deciding whether $r \geq r^*$ for any $r$) now becomes deciding whether $\max_{t \in P} d_{r} (s, t) \leq \lambda$. The algorithm of Chan and Skrepetos~\cite{ref:ChanAl16}, the algorithm of Wang and Xue~\cite{ref:WangNe20}, and the algorithm of Wang and Zhao~\cite{ref:WangAn21} are actually for finding shortest paths from $s$ to all vertices of $G_r(P)$. Thus we can solve the decision problem by using the algorithm of Chan and Skrepetos~\cite{ref:ChanAl16} for the $L_2$/$L_1$ unweighted case, the algorithm of Wang and Xue~\cite{ref:WangNe20} for the $L_2$ weighted case, and the algorithm of Wang and Zhao~\cite{ref:WangAn21} for the $L_1$ weighted case in the same way as before but with an additional last step to compute the value $\max_{t \in P} d_{r} (s, t)$ (the total running times do not change asymptotically). As such, to compute $r^*$, we can follow the same algorithm scheme as before but instead use the above new decision algorithm. In addition, for the $L_2$ unweighted case, we make the following changes to the first algorithm (the second algorithm is changed accordingly). After the $i$-th step of the BFS, which computes a set $S_i$ along with an interval $(r_1,r_2]$. If all points of $P$ have been discovered after this step and $i\leq \lfloor \lambda \rfloor$, then we have $r^*=r_2$ and stop the algorithm; the proof is similar to Lemma~\ref{lem:30}. We also stop the algorithm with $r^*=r_2$ if $i=\lfloor \lambda \rfloor$ and not all points of $P$ have been discovered; the proof is similar to Lemma~\ref{lem:35}. As before, the algorithm will stop in at most $\lfloor \lambda \rfloor$ steps. In this way, the first algorithm can compute $r^*$ in $O(\lfloor \lambda \rfloor \cdot n \log n)$  time. Analogously, the second algorithm can compute $r^*$ in $O(n^{5/4} \log^{7/4} n)$ time. For the $L_2$ weighted case, our original algorithm terminates once
$t$ is reached but now we instead halt the algorithm once all points of $P$
are reached, which does not affect the running time asymptotically.
As such, the ``single-source'' version of the $L_2$ weighted RSP problem can be solved in $O(n^{5/4}\log^{5/2}n)$ time. The $L_1$ unweighted/weighted case can be solved in $O(n\log^3 n)$
time.


\footnotesize
\bibliographystyle{plain}
\bibliography{reference}

\end{document}